\pdfoutput=1
\PassOptionsToPackage{pdftex,pdfversion=1.7,pdfencoding=auto,pdfnewwindow=true,pdfusetitle=true,pdftoolbar=false,pdfmenubar=false,bookmarks=false,bookmarksnumbered=true,bookmarksopen=false,pdfpagemode=UseThumbs,bookmarksopenlevel=1,pdfpagelabels=true}{hyperref}
\documentclass[english,superscriptaddress,twocolumn,twoside,pra,nofootinbib,colorlinks,floatfix,
noshowpacs,pdfusetitle]{revtex4-1}
\PassOptionsToPackage{usenames,dvipsnames}{color}
\usepackage{amssymb}

\usepackage[utf8]{inputenx}
\usepackage{upquote}
\usepackage{babel}
\usepackage{microtype} 
\usepackage{verbatim} 

\usepackage[usenames,dvipsnames]{color}
\definecolor{ultramarine}{RGB}{63, 0, 255}
\definecolor{medblue}{RGB}{0, 0, 100}
\definecolor{panblue}{RGB}{0,24,150}
\definecolor{carmine}{RGB}{150, 0, 24}
\definecolor{gray}{RGB}{150, 150, 150}

\usepackage{hyperref}
\hypersetup{linkcolor=medblue,citecolor=carmine,urlcolor=panblue,anchorcolor=OliveGreen}

\usepackage{units}
\usepackage{mathtools} 
\DeclarePairedDelimiter\ceil{\lceil}{\rceil}

\DeclarePairedDelimiter\parens{\lparen}{\rparen}
\DeclarePairedDelimiter\abs{\lvert}{\rvert}
\DeclarePairedDelimiter\braces{\lbrace}{\rbrace}
\DeclarePairedDelimiter\bracks{\lbrack}{\rbrack}
\newcommand{\brackets}[1]{\braces*{#1}}
\newcommand{\Tr}[1]{\operatorname{\mathsf{Tr}}\!\bracks[\big]{#1}}
\newcommand{\CH}[1]{\operatorname{\mathsf{ConvexHull}}\!\bracks[\big]{#1}}

\usepackage{braket}
\newcommand*{\dyad}[2]{ {\ket{#1}\hspace{-0.5ex}\bra{#2}}}

\usepackage{relsize}
\mathchardef\mhyphen="2D
\newcommand{\sbell}[3]{\ensuremath{{(#1\mathclose{\mhyphen}#2\mathclose{\mhyphen}#3)}}}

\usepackage{placeins} 

\definecolor{purple}{RGB}{128,0,128}

\usepackage{dsfont}
\newcommand {\id}[0]{\ensuremath{\mathds{1}}}

\newcommand{\half}[1]{\nicefrac{#1}{2}}
\usepackage{tabularx}
\usepackage{booktabs} 
\newcolumntype{R}{>{\raggedleft\arraybackslash}X}
\newcolumntype{C}{>{\centering\arraybackslash}X}
\newcolumntype{L}{>{\raggedright\arraybackslash}X}

\newcommand{\ceq}[1]{Eq.\,(\ref{#1})}
\newcommand{\fig}[1]{Fig.\,\ref{#1}}

\newcommand{\abslambda}{\ensuremath{|\lambda|}}
\newcommand{\abslambdastar}{\ensuremath{|\lambda^{\!\star}\!|}}
\newcommand{\QCN}{\ensuremath{\abs{\lambda^{\star}_{\mathcal{Q}}}_d}}

\usepackage[raggedright,bf]{subfigure}
\usepackage{aliascnt}
\usepackage{amsthm}
\newtheoremstyle{lemeq}
{}
{}
{}
{}
{\bfseries}
{:}
{\parindent}
{\thmname{#1} \thmnumber{\normalfont{#2}} \thmnote{\normalfont#3}}
\theoremstyle{lemeq}
\newtheorem{theo}{Theorem}

\newtheorem{prop}{Proposition} 
\newtheorem{axiom}{Axiom} 
\newtheorem{cor}{Corollary} 

\newtheorem*{defin}{Definition}
\newtheorem{conj}{{\color{red}Conjecture}}

\usepackage[capitalise]{cleveref} 

\Crefname{theo}{\textbf{Thm}.}{\textbf{Thms}.}
\Crefname{axiom}{\textbf{Ax}.}{\textbf{Axs}.}
\Crefname{prop}{\textbf{Prop}.}{\textbf{Props}.}
\Crefname{conj}{\textbf{Conj}.}{\textbf{Conjs}.}
\Crefname{lemma}{\textbf{Lem}.}{\textbf{Lems}.}
\Crefname{defin}{\textbf{Def}.}{\textbf{Defs}.}
\Crefname{cor}{\textbf{Cor}.}{\textbf{Cors}.}

\begin{document}

\title{Identifying Nonconvexity in the Sets of Limited-Dimension Quantum Correlations}

\author{John Matthew Donohue}
\email{jdonohue@uwaterloo.ca}
\affiliation{Institute for Quantum Computing and Department of Physics \& Astronomy, University of Waterloo, Waterloo, Ontario, Canada, N2L 3G1}

\author{Elie Wolfe}
\email{ewolfe@perimeterinstitute.ca}
\affiliation{Perimeter Institute for Theoretical Physics, Waterloo, Ontario, Canada, N2L 2Y5}


\begin{abstract}
Quantum theory is known to be nonlocal in the sense that separated parties can perform measurements on a shared quantum state to obtain correlated probability distributions which cannot be achieved if the parties share only classical randomness. Here we find that the set of distributions compatible with sharing quantum states subject to some sufficiently restricted dimension is neither convex nor a superset of the classical distributions. We examine the relationship between quantum distributions associated with a dimensional constraint and classical distributions associated with limited shared randomness. We prove that quantum correlations are convex for certain finite dimension in certain Bell scenarios and that they sometimes offer a dimensional advantage in realizing local distributions. We also consider if there exist Bell scenarios where the set of quantum correlations is never convex with finite dimensionality.
\end{abstract}
\maketitle

\section{Introduction}

A poignant illustration of the non-classicality of quantum mechanics is found in Bell scenarios~\cite{bell1964einstein}, where space-like-separated parties perform local measurements on a shared resource, as shown in \cref{fig:BellScenario}.  If the shared resource is an entangled quantum state, the conditional probability distribution resulting from the parties' measurements may be outside of the set of distributions achievable by sharing classical randomness or, equivalently, through a local hidden variable model (LHVM)~\cite{Brunner2013Bell,PopescuReviewNatureComm}.  Per the device-independent formalism, we refer to any conditional probability distribution $p[abc...|xyz...]$ as a ``box''~\cite{NoSigPolytope,PRUnit,NoSigPolytope2,GisinFramework2012,ScaraniNotes2,scarani2012device,BancalDIApproach,GPTrelatingDI}, where the parties' choices of measurements $x,y,z...$ are referred to as ``inputs'' and their measurement results $a,b,c...$ as ``outputs''.  In this framework the quantum nature of the resource is assessed without information regarding the internal mechanism of the box, but rather by determining whether or not the box belongs to the LHVM-compatible set of boxes.  Nonlocal boxes are those which cannot be implemented with classical shared randomness, and are highly valued as the resources which drive device-independent quantum communication and quantum cryptography~\cite{CryptoDIQKDNJP,masanesDIQKD,RenatoQKDNature}.

For any given Bell scenario, the set of LHVM-compatible boxes is known as the local polytope; it has a finite number of facets given by Bell inequalities~\cite{Brunner2013Bell}.  A box is defined as nonlocal if and only if it violates a Bell inequality.  By contrast, we refer to the set of boxes implementable by sharing any quantum state as the quantum elliptope, which is also convex~\cite{WernerWolfBellInequalities,FritzCombinatorialLong,KirchbergConjecture,TsirelsonProblemInteraction} but cannot be described in terms of linear inequalities~[\citealp{Avis2007}, see \citealp[Fig. 4]{Brunner2013Bell}].  Boxes inside the quantum elliptope have conditional probabilities of the form \begin{align}\label{eq:genericquantumbox}
\hspace{-1ex}p[abc...|xyz...]=\Tr{\rho \, \parens[\big]{\hat{A}_{a|x} \otimes \hat{B}_{b|y} \otimes \hat{C}_{c|z} \,...}},
\end{align} where the dimensions of $\rho$ and the local measurement operators $\hat{A}_{a|x},\hat{B}_{b|y},\!...$ are unconstrained.  To enforce no-signalling, i.e. to ensure that no information about the measurement choices of one party can be inferred by the measurement results of the other, the parties are assigned distinct Hilbert spaces~\cite{FritzCombinatorialLong,AlmostQuantum}.  The operational boundary conditions delineating the quantum elliptope are notoriously hard to pin down, although various approximations to the quantum elliptope have recently become available \cite{AlmostQuantum,NPA2007Short,NPA2008Long,NPAReview}.  The quantum elliptope strictly contains the local polytope; nevertheless boxes inside it cannot be used for nonlocal signalling~\cite{Brunner2013Bell}.

\begin{figure}[t!]
\includegraphics[width=1\columnwidth]{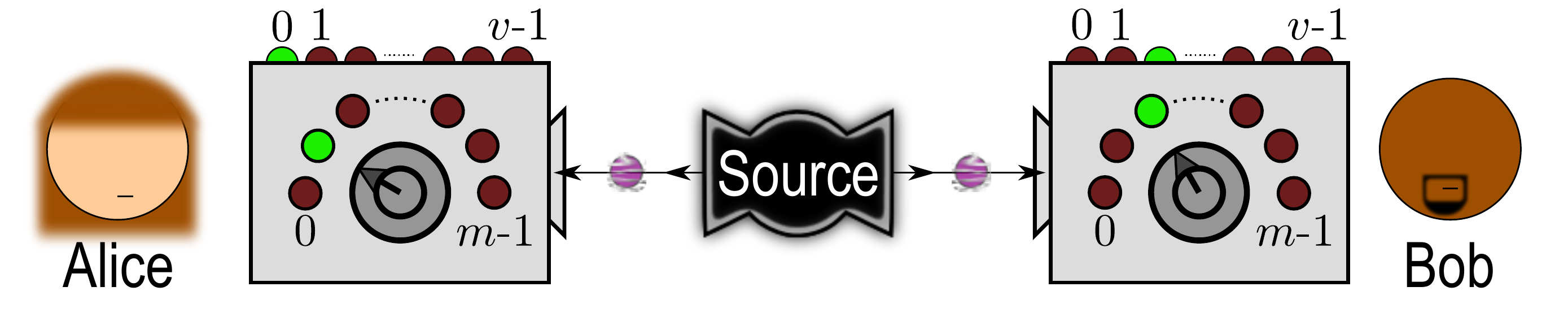}
\vspace{-6ex}
\caption{Bell scenarios may be viewed schematically as shown above, where $n$ parties ($n=2$ shown in this example) make synchronized measurements on systems prepared by a common source. Each party has $m$ input (measurement setting) choices and $v$ output (measurement result) possibilities. For $n=2$, the parties are conventionally named Alice and Bob. Inputting $x$ and $y$ will return outputs $a$ and $b$ for Alice and Bob, respectively, with probability $p[ab|xy]$.}\label{fig:BellScenario}\vspace{-3ex}
\end{figure}

In this letter, we show that convexity and containment of the local polytope are lost when the dimensions of the local Hilbert spaces are restricted. Although dimensionally constrained quantum correlations have already garnered considerable attention, most prior works have considered the parties as additionally having access to unlimited shared randomness. Examples include determining the degree to which Bell inequalities can be violated \cite{DimensionDependantBellViolation,DimensionDependantBellViolation2,TsirelsonBoundSVD}, assessing the security of a quantum cryptographic protocol \cite{Nqubit.qb,QKD.DD.Winter,PhysRevLett.95.080501}, and testing if a given probability distribution is achievable \cite{DimensionWitness2008Brunner,Navascues2013FiniteDimensions,Navascues2014FiniteDimensions}. The background assumption of unlimited shared randomness imposes convexity automatically, and obscures the fine-grained quantum-to-classical comparison which we seek.  In contrast, our approach considers purely quantum systems or systems with limited shared randomness.   Foundationally, the non-classical nature of quantum mechanics is richer in detail if we compare quantum preparations to classical preparations at finite sizes, as opposed to only in the asymptotic limits of the quantum elliptope and local polytope~\cite{GPTrelatingDI}. The aim of this Letter is to explore this finer characterization. 

Some important work has already been done in this framework. For example, \citet{ClassicalDimensionBoundNonConvex} recently re-analyzed prepare-and-measure signalling scenarios without the assumption of unlimited shared randomness~\cite{ClassicalDimensionBound2010}.  The minimal classical or quantum dimension required to achieve some \emph{unconditional} joint probability distribution, i.e. limiting each party to a single input, has recently been equated with the distribution's non-negative or positive semi-definite rank, respectively~\cite{PSDasQuantum,ReviewPSD}, and an arbitrary strong quantum dimensional advantage has been noted~\cite{ZhangQuantumGameTheory}.  For the usual multiple-inputs Bell scenario, the question of minimizing the quantum dimension required to implement a given box has recently been considered in Ref.~\cite{ConditionalMinimumHSD}, though see also Refs.~\cite{ClassicalLambdaRequirement,DimensionWitness2008Wehner}.  \citet{ConcavityPaperVertesi} notably called attention to the nonconvexity of some dimensionally constrained quantum correlations. We resolve the open question they posed by showing that sharing merely qubits leads to a non-convex set of correlations in \emph{any} non-trivial Bell scenario.  We then extend the discussion to multipartite and multichotomous measurements.  We prove that, in some scenarios, quantum correlations are guaranteed to be convex if the underlying Hilbert space dimension is sufficiently, but finitely, large.

\section{Notation, definitions, and fundamental hierarchies}\label{sec:notation}

We shall examine Bell scenarios where $n$ space-like separated parties choose from among $m$ inputs (measurement settings) each and, in response, observe among $v$ possible outcomes (measurement results),  as shown in \cref{fig:BellScenario}.  We label such Bell scenarios with the index \sbell{n}{m}{v}.  The \sbell{n}{2}{2} scenarios, for example, are those for which every party has access to two binary-outcome observables, and \sbell{2}{2}{2} is the familiar CHSH nonlocality scenario~\cite{CHSHOriginal,WolfeQB}.  For analytical clarity, we consider exclusively symmetric Bell scenarios.  As is conventional, for \sbell{2}{m}{v} the two parties are referred to as Alice and Bob. We use $x$ and $y$ to indicate Alice and Bob's respective apparatus choices (box inputs), and use $a$ and $b$ to indicate their respective measurement outcomes (box outputs), starting all indices from $0$.

\begin{figure*}[t!]
\centering
\subfigure[Sharing a multipartite quantum state.\hfill]
{\begin{minipage}[t]{.3\textwidth}
    \includegraphics[width=0.75\linewidth]{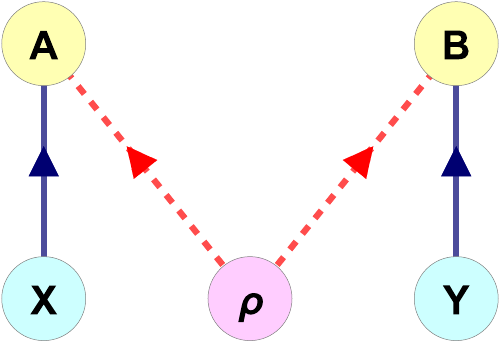}
    \label{fig:qDAG}
\end{minipage}}\hfill{\color{gray}\rule[-0.4cm]{0.5pt}{3cm}}\hfill
\subfigure[Sharing only classical randomness.\hfill]
{\begin{minipage}[t]{.3\textwidth}
    \includegraphics[width=0.75\linewidth]{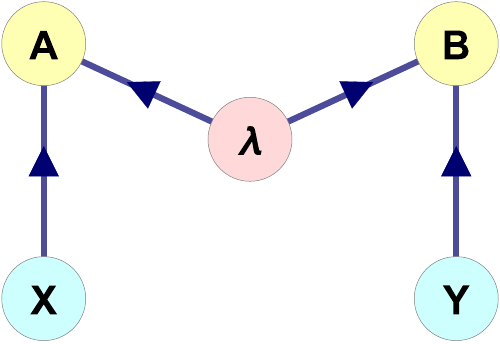}
    \label{fig:lDAG}
\end{minipage}}\hfill{\color{gray}\rule[-0.4cm]{0.5pt}{3cm}}\hfill
\subfigure[Sharing both resources.]
{\begin{minipage}[t]{.3\textwidth}
    \includegraphics[width=0.75\linewidth]{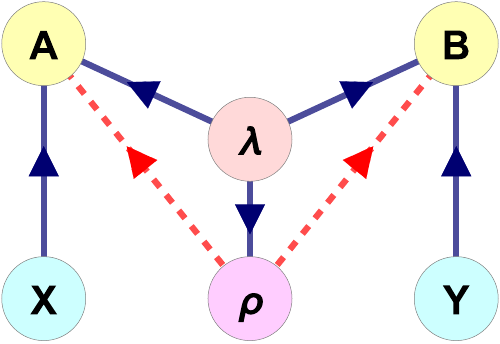}
    \label{fig:lqDAG}
\end{minipage}}
\vspace{-3ex}\caption{Various causal structures associated with two space-like separated parties who make synchronized measurements on systems prepared by a common source, as in \cref{fig:BellScenario}. The joint distribution on the observed outcomes differs depending on the nature of the shared resource. If the preparation is quantum, such as in \subref{fig:qDAG}, then $p[ab|xy]=\Tr{\rho \, \hat{A}_{a|x} \otimes \hat{B}_{b|y}}$, per \ceq{eq:genericquantumbox}. If the preparation is classical, such as in \subref{fig:lDAG}, then $p[ab|xy]=\sum_{\lambda}p[\lambda]p[a|x\lambda]p[b|y\lambda]$, per \ceq{eq:genericclassicalbox}. If the shared system is both classical and quantum, such as in \subref{fig:lqDAG}, then $p[ab|xy]=\sum_{\lambda}p[\lambda]\Tr{\rho^{(\lambda)} \, \hat{A}^{(\lambda)}_{a|x} \otimes \hat{B}^{(\lambda)}_{b|y} }$, per \ceq{eq:generichybridbox}. Hybrid preparations can equivalently be considered classical-quantum (cq) states \cite{CQstate2008,DimensionWitness2008Wehner}.
}
\label{fig:DAGs}\vspace{-2ex}
\end{figure*}

To avoid over-specifying a no-signalling box, we parameterize boxes such that only the minimum number of conditional probabilities are specified to allow full reconstruction of the distribution.  One such parameterization (inspired by \citet{ExplicitNoSigBasis}) is based on reserving-as-implicit all probabilities involving the outcome $0$.  For example, since ${\sum_{b=0}^{v-1}{p[ab|xy]=p[a|x]}}$ by no-signalling, one finds that ${p[a 0|xy]=p[a|x]-\sum_{b=1}^{v-1}{p[ab|xy]}}$.  In this parameterization scheme, each of the $\binom{n}{k}m^k$ possible $k$-partite input tuples one might condition upon has $(v-1)^k$ freely specifiable output probabilities. Consequently, any no-signalling box is specified in terms of precisely
\begin{align}\label{eq:dimcount}
\hspace{-2ex}\mathcal{F}=\sum_{k=1}^{n}{\binom{n}{k}m^k{(v-1)}^k}={{\left(m(v-1)+1\right)}^n-1}
\end{align}
total explicit parameters, regardless of the choice of parameterization scheme \cite{PironioDimension,Brunner2013Bell}.

For further simplification, we will only consider quantum measurement scenarios for which the local Hilbert spaces of every party have the same dimension\footnote{The uniform-dimension constraint is insignificant in bipartite scenarios; see \cref{sec:asymmdim} for details.}, $d$. Let's denote the set of quantum boxes achievable with such a dimensional constraint as $\mathcal{Q}_d^{\sbell{n}{m}{v}}$.  In this notation, therefore, the quantum elliptope is identically $\mathcal{Q}_{\infty}^{\sbell{n}{m}{v}}$.  In this set, the correlation between the parties arises entirely from measurements on a shared quantum state, as shown in \cref{fig:qDAG}.  We'll indicate $\mathcal{Q}_d^{\sbell{n}{m}{v}}$ via the shorthand $\mathcal{Q}_d$ when the scenario \sbell{n}{m}{v} is clear from context.  

It will also be important for us to consider the set of boxes achievable by sharing only a classical random variable $\lambda$ (as opposed to sharing a quantum state $\rho$), as seen in \cref{fig:lDAG}. In parallel with the notion of constrained local Hilbert space dimension, we let $\mathcal{L}$ denote the set of boxes achievable by sharing only a classical random variable $\lambda$ (as opposed to sharing a quantum state $\rho$).  The set of classical boxes achieved using \emph{constrained} shared randomness is analogously denoted $\mathcal{L}_{\abslambda}$, where $\abslambda$ refers to the dimension of the shared classical randomness, i.e. $\mathcal{L}_{\abslambda}$ indicates that the parties share no more than $\log_2\abslambda$ classical bits in common. For example, if the parties share the outcomes of rolling two distinguishable and arbitrarily weighted dice, then $\abslambda=36$. 

Boxes predicated on classical randomness are identically mixtures of product distributions\footnote{The notion of sharing classical randomness can be generalized to random variables with ``memory", but our analysis hews to the simplest model; see Ref. \citep[Sec. II.G]{Brunner2013Bell} for more details.}, i.e. $\mathcal{L}_{\abslambda}$ is the set of boxes for which
\begin{align}\label{eq:genericclassicalbox}
p[ab...|xy...]=\sum_{\lambda=0}^{\abslambda-1}p[\lambda]p[a|x\lambda]p[b|y\lambda]...
\end{align}
Accordingly, $\mathcal{L}_{\infty}$ may be thought of as the convex hull of $\mathcal{L}_{1}$, where $\mathcal{L}_{1}$ is the set of boxes achievable without actually sharing any randomness. Indeed, that $\mathcal{L}_{\infty}=\CH{\mathcal{L}_{1}}$ is a corollary of Fine's theorem~\cite{FineTheorem,SpekkensDeterminism,KunjwalFineTheorem}.  

While $\mathcal{L}_{\infty}=\CH{\mathcal{L}_{1}}$ follows from \cref{eq:genericclassicalbox}, one can nevertheless span the local polytope without requiring infinite shared randomness. We define the minimum amount of shared randomness which allows every possible local box to be implemented as $\abslambdastar$:
\begin{defin}[$\abslambdastar$]
The classical Carathéodory number $\abslambdastar$ is the minimum amount of classical randomness that must be shared in order to span the local polytope. Formally,
\begin{align}\label{eq:minsharedrandom}
\abslambdastar\equiv \min \abslambda \text{ s.t. }\mathcal{L}^{\sbell{n}{m}{v}}_{\infty}=\mathcal{L}^{\sbell{n}{m}{v}}_{\abslambda}
\end{align}
where $\abslambdastar$ intrinsically depends implicitly on \sbell{n}{m}{v}.
We call $\abslambdastar$  the classical Carathéodory number \cite{Caratheodory,CaratheodoryMathworld} of $\mathcal{L}_{1}$ because every $P\in\CH{\mathcal{L}_{1}}$ can be decomposed as a convex mixture of at-most $\abslambdastar$ boxes from $\mathcal{L}_1$, whereas $\abslambdastar-1$ would not be adequate.  Details regarding the determination of $\abslambdastar$, including the result that $\abslambdastar=4$ for the \sbell{2}{2}{2} scenario, can be found in \cref{sec:lambdastar}~\cite{StirlingS2Mathworld}.\end{defin}

Note that boxes arising from measurements on \emph{separable} quantum states are equivalent to boxes arising from shared classical randomness and vice versa. Critically, the set of all classical boxes with degree of shared classical randomness  $\abslambda$  is contained by the set of boxes implemented quantumly by sharing separable qudits with local Hilbert space dimension $d\geq\abslambda$. Formally,
\begin{axiom}\label{theo:LinsideQ}
$\mathcal{L}_{\abslambda}\subseteq{(\mathcal{Q}:\mbox{sep})}_{\mathlarger{d} \geq \abslambda},\,$ i.e. any box that can be constructed by sharing a classical dit can also be composed by sharing qudits.
\end{axiom}
\begin{cor}
\label{cor:PolytopeinsideQ}
$\mathcal{L}_{\infty}\subseteq\mathcal{Q}_{d}$ whenever $\abslambdastar\leq d$.
\end{cor}
\begin{cor}
\label{cor:SubLocalNonConvex}
If $\mathcal{L}_{\infty}\nsubseteq\mathcal{Q}_{d}\,$, then $\mathcal{Q}_{d}$ is not convex.
\end{cor}
\begin{proof}
Every $P\in \mathcal{L}_{\abslambda}$ can be mapped to a $P^{\prime}\in {(\mathcal{Q}:\mbox{sep})}_{d=\abslambda}$ by the following construction on \cref{eq:genericquantumbox}: $\rho\to{\sum_{\lambda}^{\abslambda}p[\lambda]\parens*{\dyad{\lambda}{\lambda}}^{\otimes n}}$,  ${\hat{A}_{a|x}\to \sum_{\lambda}^{\abslambda}p[a|x\lambda]\dyad{\lambda}{\lambda}}$, and so on. Thus, for example, any box in $\mathcal{L}_{1}$ can be implemented quantumly by sharing a product state.  Note that the shared randomness required to enclose the separable quantum set of dimension $d$ is $\abslambda\geq d^n$, as shown in \cref{sec:compaxiom}~\cite{NSepStateMixing,SeparableEnsembleLength}. 

The corollaries follow since $\mathcal{L}_{\infty}=\mathcal{L}_{\abslambdastar}=\CH{\mathcal{L}_{1}}$.  It follows from \cref{eq:minsharedrandom,theo:LinsideQ} that the local polytope is contained by the quantum set with Hilbert space dimension $d\geq\abslambdastar$.  Since $\mathcal{L}_{\infty}=\CH{\mathcal{L}_{1}}$, we conclude that $\mathcal{Q}_{d}$ is not convex if it does not contain the local polytope, i.e. whenever ${\mathcal{L}_{\infty}\nsubseteq\mathcal{Q}_{d}}$.  Therefore, evidence of any LHVM boxes not achievable by sharing quantum states is evidence of nonconvexity.
\end{proof}

We conjecture that the shared quantum state must have local dimension $d\geq \abslambdastar$ in order to contain the local polytope, i.e. the converse of \cref{cor:PolytopeinsideQ}, that ${\mathcal{L}_{\infty}\nsubseteq\mathcal{Q}_{d<\abslambdastar}}$. For that matter, we conjecture that the converse of \cref{theo:LinsideQ} is also true, but this is a fundamental unproven open question.

\begin{conj}\label{conj:LnotsubsumedbysmallQ}
$\mathcal{L}_{\abslambda\leq\abslambdastar}\nsubseteq{\mathcal{Q}}_{\mathlarger{d} < \abslambda},\,$ i.e. the set of boxes which may be realized by sharing (possibly entangled) qudits of local Hilbert space dimension $d$ is conjectured to \emph{not} contain the set of boxes constructable by sharing a classical random variable of dimension $\abslambda\leq\abslambdastar$ if $d<\abslambda$.
\end{conj}

\begin{figure*}[!t]
\centering
\subfigure[If $\;\widetilde{P}=c_1P_1\!+c_3P_3\!+c_4P_4,$
\newline then $\;\widetilde{P}\in\mathcal{L}_{\infty}\;$ iff $\;c_4\leq 1-c_3,$
\newline and $\;\,\,\widetilde{P}\in\mathcal{Q}_2\;\;$ iff $\;\brackets{c_1 \text{ or } c_3\text{ or }c_4}=0.$
]{\begin{minipage}[t]{.3\textwidth}
\includegraphics[width=0.8\linewidth]{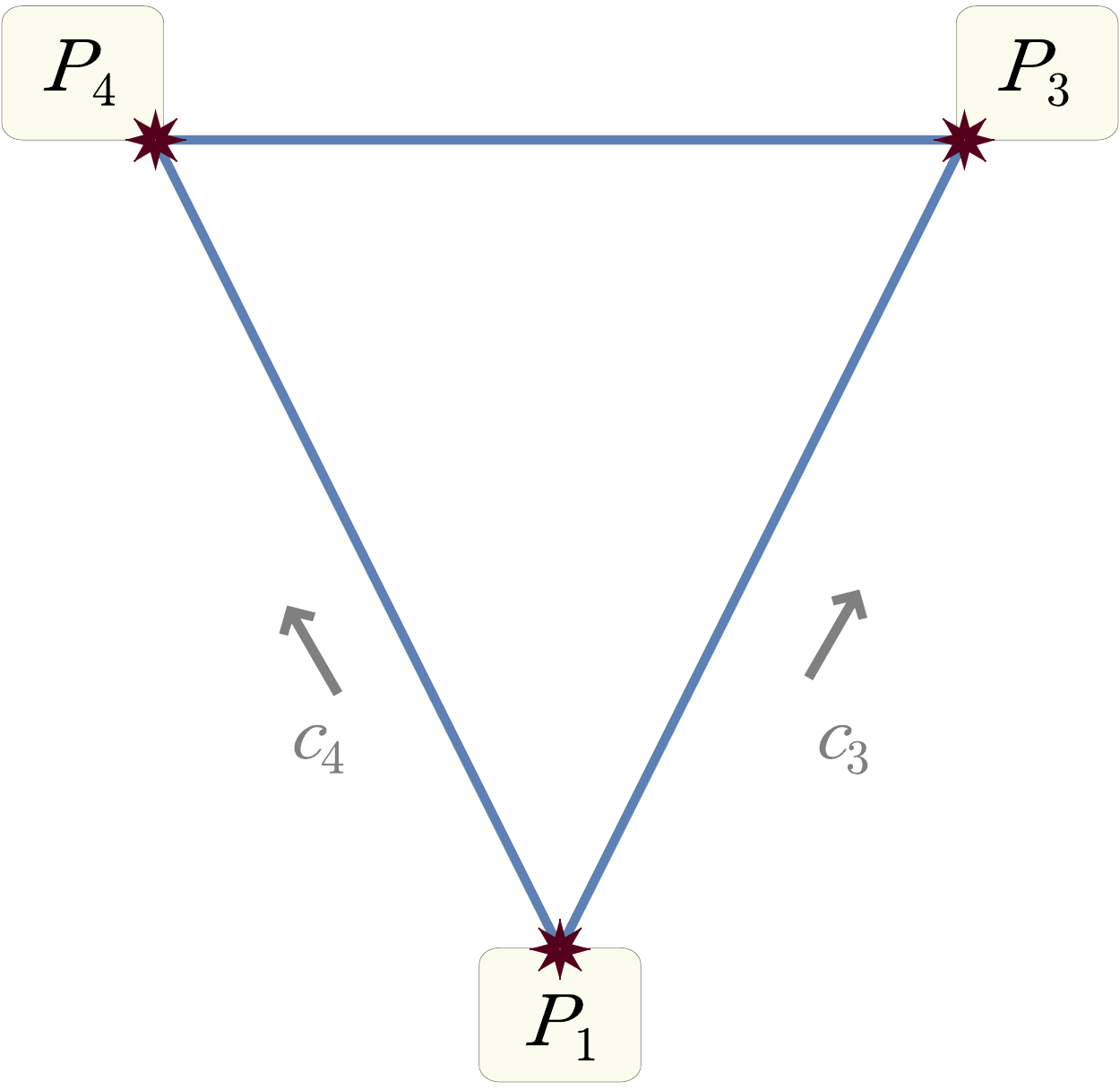}\label{fig:Tri134}\vspace{1ex}
\end{minipage}}
\hfill{\color{gray}\rule[-0.75cm]{0.5pt}{5.5cm}}\hfill
\subfigure[If $\;\widetilde{P}=c_0P_0\!+c_1P_1\!+c_{3:4}P_{3:4},$
\newline then $\;\widetilde{P}\in\mathcal{L}_{\infty}\;$ iff $\;c_{3:4}\leq 1-c_{1},$
\newline and $\;\,\,\widetilde{P}\in\mathcal{Q}_2\;\;$ iff $\;c_{3:4}\leq {\left(1-c_{1}\right)}^{\nicefrac{3}{2}}.$
]{\begin{minipage}[t]{.3\textwidth}
\includegraphics[width=0.8\linewidth]{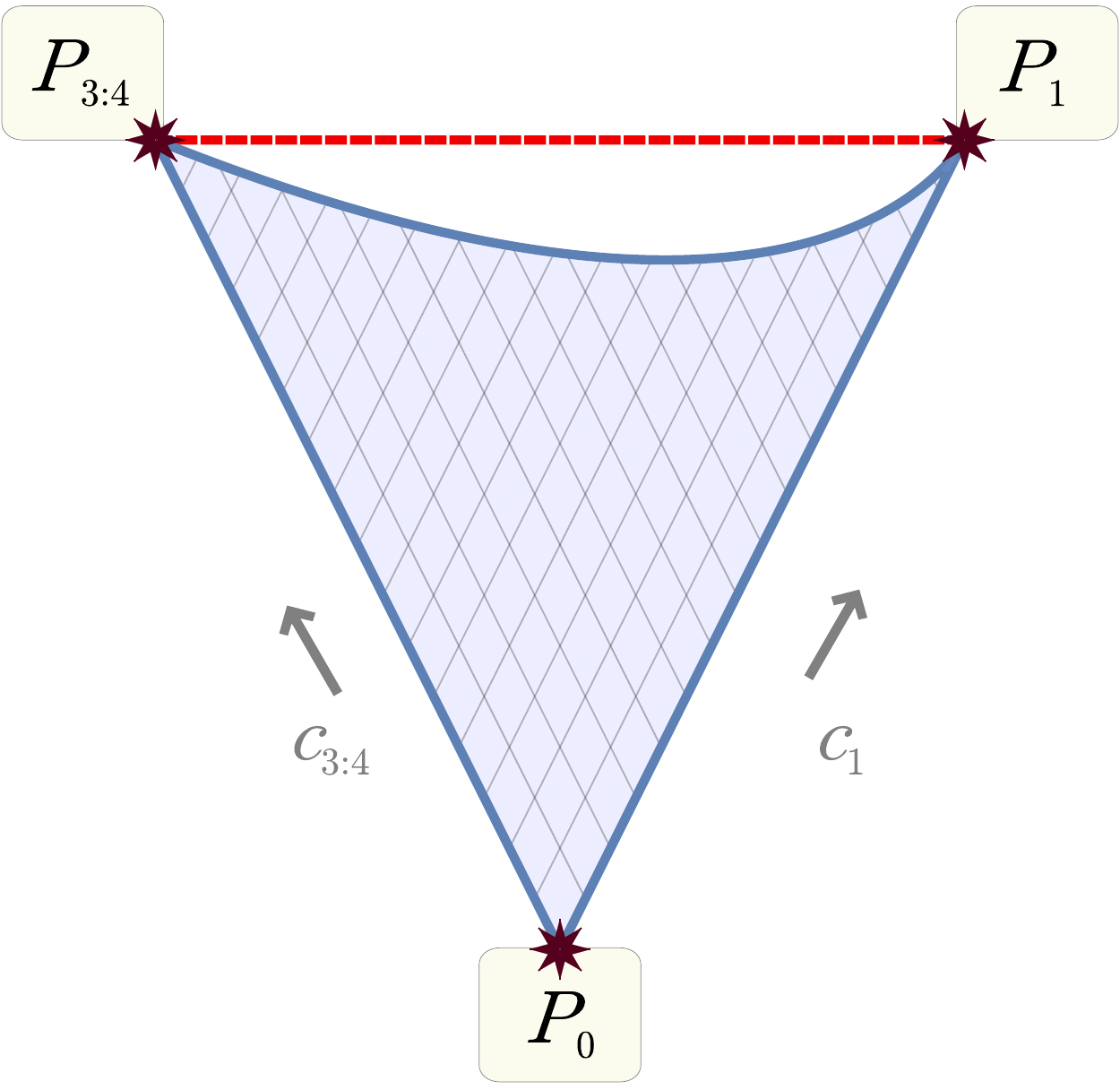}\label{fig:Tri0134}\vspace{1ex}
\end{minipage}}
\hfill{\color{gray}\rule[-0.75cm]{0.5pt}{5.5cm}}\hfill
\subfigure[If $\;\widetilde{P}=c_{0}P_{0}\!+c_{1}P_{1}\!+c_{\text{TB}}P_{\text{TB}},$
\newline then $\;\widetilde{P}\in\mathcal{L}_{\infty}\;$ iff $\;c_{\text{TB}}\leq 2^{\half{-1}}\left(1-c_{1}\right)\!,$
\newline and $\;\,\,\widetilde{P}\in\mathcal{Q}_2\;\;$ iff $\;c_{\text{TB}}\leq {\left(1-c_{1}\right)}^{\nicefrac{5}{4}}.$
]{
\begin{minipage}[t]{.3\textwidth}
\includegraphics[width=0.8\linewidth]{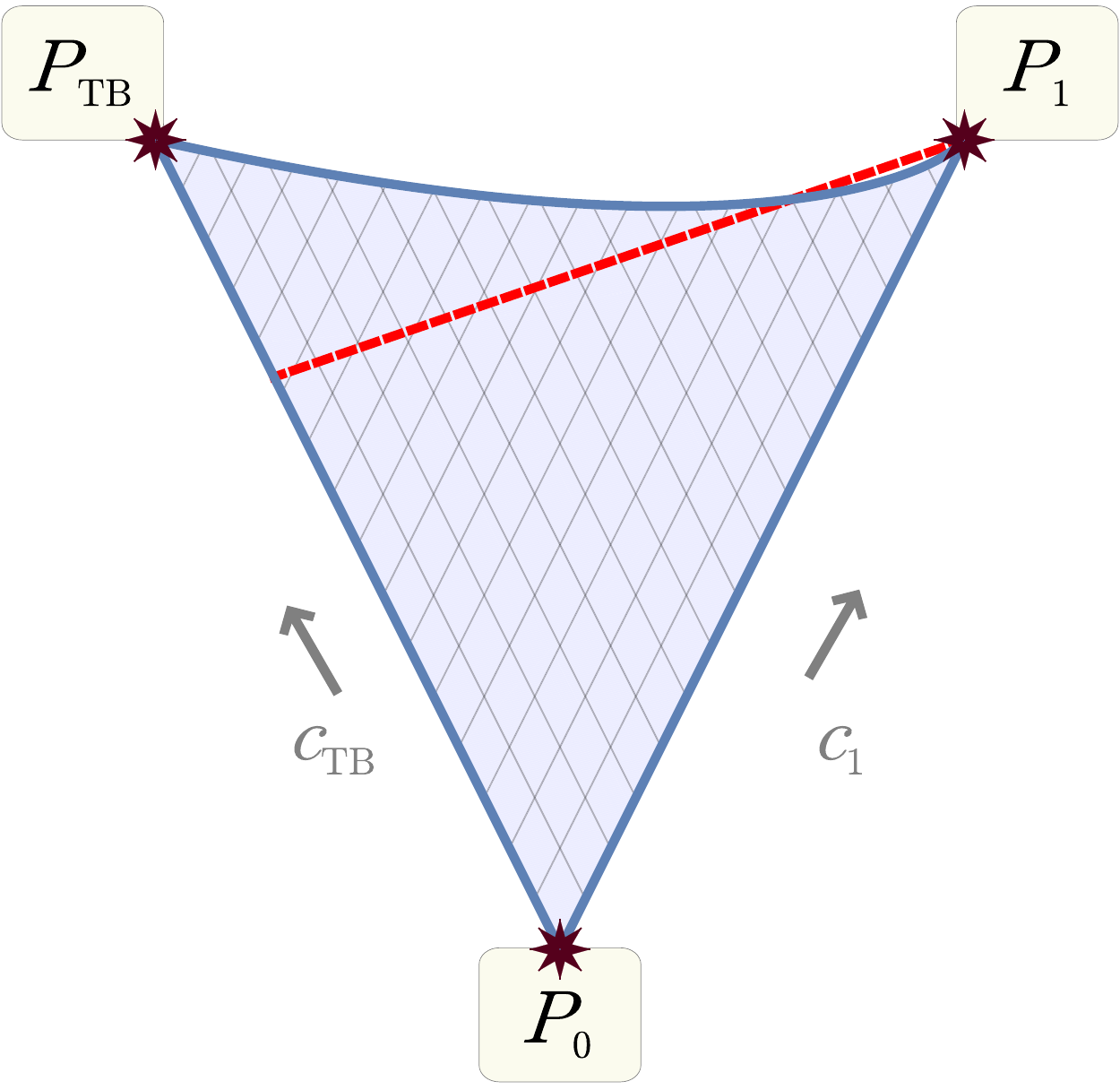}\label{fig:Tri01TB}\vspace{1ex}
\end{minipage}}
\vspace{-2ex}\caption{Three examples of the nonconvexity of the set $\mathcal{Q}_{2}$ are shown above. Each triangle represents the set of possible convex combination of its vertex boxes as defined in \cref{tab:local4}, where $c_i$ represents the weight of $P_i$ in the mixture and $\sum_{i}{c_i}=1$.  To emphasize the notion of convex combination we have plotted the regions as equilateral triangles; nevertheless any pair of edges should be thought of as two independent axes such that the third weight is fixed by normalization.  Although in each triangle all the vertex boxes are achievable with shared-qubits preparations, only a fraction of possible convex combinations are also achievable, as indicated by the shaded region. Indeed, in \subref{fig:Tri134} only the zero-area edges of the triangle are qubit-achievable. The local polytope $\mathcal{L}_{\infty}$ contains the entire triangular regions of \subref{fig:Tri134} and \subref{fig:Tri0134} since the vertex boxes of those triangle are all local, and in \subref{fig:Tri01TB} only the region below the dashed red line is consistent with classical shared randomness. The essential feature of all three triangles is that the shaded subset representing $\mathcal{Q}_{2}$ is not convex and that $\mathcal{L}_{\infty}\nsubseteq\mathcal{Q}_{2}$. \qquad The $\mathcal{Q}_{2}$ region in \subref{fig:Tri01TB} was obtained via computational maximization; ${\left(1-c_1\right)}^{\nicefrac{5}{4}}$ was noticed to precisely coincide with the numeric boundary. The given boundaries of $\mathcal{Q}_{2}$ in \subref{fig:Tri134} and \subref{fig:Tri0134} were recovered analytically. \qquad  Consistent with Ref. \cite{ConcavityPaperVertesi}, the qubit-accessible region shrinks if one restricts to projective measurements. The edges where $c_1\neq 0$ in \subref{fig:Tri134}, for example, are not in $\mathcal{Q}_2$-PVM, and in \subref{fig:Tri01TB} we find that  $\widetilde{P}\in\mathcal{Q}_2$-PVM only if $c_{\text{TB}}\lesssim {\left(1-c_{1}\right)}^{\nicefrac{3}{2}}$. Thus the use of POVMs can be certified in a device-independent way given a dimension promise, akin the the results of Ref. \cite{VBInequality}.
}\label{fig:triangles}\vspace{-2ex}
\end{figure*}

Note that \cref{conj:LnotsubsumedbysmallQ} would imply the nonconvexity of $\mathcal{Q}_d$ whenever $d<\abslambdastar$, per \cref{cor:SubLocalNonConvex}. One can use entropic measures to infer a trivial lower bound on the smallest $d$ for which $\mathcal{L}_{\abslambda}\subseteq\mathcal{Q}_{d}$ by comparing the largest mutual information which can be mediated through the given $\abslambda$ to the maximum total correlation capacity of quantum states of dimension $d$ \citep[Eq. (16)]{PopescuTotalCorrelation}.

\section{Non-convex sets of quantum correlations}

\begin{table}[!b]
\vspace{-2.5ex}
\centering
\begin{tabularx}{\linewidth}{|C|CCCCcccc|}
\toprule
  & \(\Braket{A_0}\) & \(\Braket{A_1}\) & \(\Braket{B_0}\) & \(\Braket{B_1}\) & \(\Braket{A_0 B_0}\) & \(\Braket{A_1 B_0}\) & \(\Braket{A_0 B_1}\) & \(\Braket{A_1 B_1}\) \\
 \midrule
 \(P_0\) & 0 & 0 & 0 & 0 & 0 & 0 & 0 & 0 \\
 \(P_1\) &  1 &  1 &  1 &  1 & 1 &  1 &  1 & 1 \\
 \(P_2\) & -1 & -1 & -1 & -1 & 1 &  1 &  1 & 1 \\
 \(P_3\) &  1 & -1 &  1 & -1 & 1 & -1 & -1 & 1 \\
 \(P_4\) & -1 &  1 & -1 &  1 & 1 & -1 & -1 & 1 \\
\midrule
 \(P_{3:4}\) & 0 & 0 & 0 & 0 & 1 & -1 & -1 & 1 \\
  \(P_{1:4}\) & 0 & 0 & 0 & 0 & 1 & 0 & 0 & 1 \\
\midrule
 \(P_{\text{TB}}\) & 0 & 0 & 0 & 0 & \(2^{\half{-1}}\) & \(2^{\half{-1}}\) & \(2^{\half{-1}}\) & \(-2^{\half{-1}}\) \\
 \bottomrule
\end{tabularx}\vspace{-1ex}
\caption{\label{tab:local4}A list of bipartite conditional probability distributions, or boxes, in the \sbell{2}{2}{2} scenario. Per convention~\cite{Brunner2013Bell}, we parameterize binary-output boxes in terms of outcome biases, i.e. $\Braket{A_x}=p[a\mathopen{=}1|x]-p[a\mathopen{=}0|x]$ and $\Braket{A_xB_y}=p[a\mathopen{=}b|xy]-p[a\mathopen{\neq}b|xy]$. Note that the five boxes $P_0$ through $P_4$ are product distributions, achievable even absent any shared randomness, i.e. $\abslambda=1$. \(P_{i:j}\), by contrast, indicates the equally-weighted mixture of boxes \(P_i\) through \(P_j\), requires non-trivial shared randomness. Only \(P_{\text{TB}}\), the quantum box which achieves the Tsirelson bound \cite{Tsirelson1980,WolfeQB}, is nonlocal. Every box in this table is achievable with qubits, but many mixtures of these boxes \emph{cannot} be achieved using qubits, per \fig{fig:triangles}.
\vspace{-1ex}
}
\end{table}

With the above definitions in place, we next demonstrate the nonconvexity of the set of quantum distributions explicitly.  For this purpose, it suffices to focus on the most trivial example imaginable: a one-dimensional quantum system, ergo the set $\mathcal{Q}_{1}$.  

\begin{prop}\label{lem:Q1notconvex}
$\mathcal{Q}_{1}$ is not convex.
\end{prop}
\begin{proof}

If ${d=1}$ then the only possible quantum state that Alice and Bob can ``share" is the product state $\ket{0}\bra{0}_A\otimes\ket{0}\bra{0}_B$.  No matter how Alice and Bob choose their POVM elements, their joint probability distributions will always factorize to a product distribution, $p[ab|xy]=\Tr{\hat{A}_{a|x}}\Tr{\hat{B}_{b|y}}=p[a|x]\times p[b|y]$.  Indeed, $\mathcal{Q}_{1}$ is the set of all product distributions, equivalent to the set the of boxes achievable with any $d$ subject to constraining $\rho$ to be a product state. As such, $\rho$ isn't really ``shared"
at all; the only randomness is local noise.  These restrictions are certainly limiting, but because all local deterministic boxes are achievable in $\mathcal{Q}_{1}$, we thus have $\mathcal{L}_{\infty}\subseteq\CH{\mathcal{Q}_{1}}$.


Suppose Alice and Bob both have access to a uniformly-distributed variable $0\leq\lambda\leq v-1$, and let their marginal probabilistic dependencies on $\lambda$ be such that $p[a|x\lambda]\to\delta[a=\lambda]$ and $p[b|y\lambda]\to\delta[b=\lambda]$ regardless of $x$ or $y$. The resulting box $P\in\mathcal{L}_{v}$ is such that ${p[\lambda\lambda|xy]=p[\lambda|x]=p[\lambda|y]=v^{-1}}$ where $P$ does not factorize, i.e. $p[\lambda\lambda|xy]\neq p[\lambda|x]\times p[\lambda|y]$. As such, $P\in\mathcal{L}_{\infty}$ yet $P\not\in\mathcal{Q}_{1}$. By \cref{cor:SubLocalNonConvex}, then, $\mathcal{Q}_{1}$ is not convex.
\end{proof}

In a less trivial example, we also show that nonconvexity persists in the set $\mathcal{Q}_{2}$ when the number of inputs $m\geq2$.  In Ref.~\cite{ConcavityPaperVertesi} the nonconvexity of $\mathcal{Q}_{2}$ was also proven, but only for scenarios with $m\geq4$.

\begin{prop}\label{lem:Q2notconvex}
$\mathcal{Q}_{2}$ is not convex (for all $m\geq 2$). \\ For example, in the \sbell{2}{2}{2} scenario, the set of boxes that can be achieved by sharing qubits is nonconvex.
\end{prop}
\begin{proof}
To show this, we considered maximally general two-qubit preparations and measurement schemes, and found that convex combinations of various boxes in $\mathcal{Q}_{2}$ listed in \cref{tab:local4} were not themselves achievable with qubits\footnote{A brute-force proof seemed necessary to the authors. The method of \citet{ConditionalMinimumHSD}, however, can also be used to prove the essential nonconvexity, by certifying that \emph{some} of the boxes in the unshaded region of \cref{fig:Tri134} necessitate a quantum description in terms of a Hilbert space dimension greater than two, i.e. that qubits would be incompatible. Brute-force feasibility-checking is nevertheless required to generate the nonconvexity illustrations in \fig{fig:triangles}.};  see \fig{fig:triangles} for illustrative examples. We initially identified a few nonconvex boundaries of $\mathcal{Q}_{2}$ via computational maximization, most of which we subsequently recovered analytically. Details of our state and measurements parametrization may be found in \cref{sec:params}~\cite{acin2000generalized,characterizingentanglement,multireview,entang.review.toth,HeinosaariPOVM,ConnesEmbedding,PhysicalTsirelson,KirchbergConjecture,TsirelsonProblemArXiv,VBInequality}. The full derivation of the boundaries indicated in \fig{fig:triangles} may be found in a {\it Mathematica$^{_{\text{\it\tiny\texttrademark}}}$} notebook in the Supplementary Online materials \cite{localsom}.
\end{proof}

\begin{prop}\label{prop:somesuperlocality}
$\left(\mathcal{L}_{\infty}\setminus{\mathcal{L}_{\abslambda=d}}\right) \cap \mathcal{Q}_{d}\neq\mathlarger{\operatorname{{\emptyset}}},\;$ i.e. there exist local boxes which classically require shared-randomness of dimension at least $|\lambda^{\prime}|$, but which admit quantum ``shortcuts" through implementations using quantum systems of smaller dimension $d<|\lambda^{\prime}|$.
\end{prop}
\begin{proof}
Consider the local box $P_{1:4}$ from \cref{tab:local4}. If Alice and Bob choose the same input then their outputs are perfectly correlated, but if they choose different inputs then their outputs are unrelated. To achieve $P_{1:4}$ classically requires flipping two coins: When Alice or Bob input $0$ their output is given by the first coin flip. The input $1$, however, returns the second coin flip. Thus $P_{1:4}$ requires $\abslambda\geq4$, i.e. $P_{1:4} \in \mathcal{L}_{\infty}\setminus{\mathcal{L}_{3}}$. Alternatively, $P_{1:4}$ is quantumly achievable by sharing the entangled pure state $\frac{\ket{00}+\ket{11}}{\sqrt{2}}$ and letting ${\hat{A}_{1|0}-\hat{A}_{0|0}}={\hat{B}_{1|0}-\hat{B}_{0|0}}={\hat{\sigma}}_{Z}$ and ${\hat{A}_{1|1}-\hat{A}_{0|1}}={\hat{B}_{1|1}-\hat{B}_{0|1}}={\hat{\sigma}}_{X}$. So, $P_{1:4} \in \mathcal{Q}_{2}$. 
\end{proof}

Although a box which requires $\abslambda\geq4$ is inside $\mathcal{Q}_{2}$, nevertheless $\mathcal{L}_{4}\nsubseteq\mathcal{Q}_{2}$. Indeed, the empty interior of \fig{fig:Tri134} shows that $\mathcal{L}_{3}\nsubseteq\mathcal{Q}_{2}$, as is expected per \cref{conj:LnotsubsumedbysmallQ}. Thus entanglement enables not only non-locality, ie. boxes forbidden classically, but also \emph{super}-locality: quantum measurement schemes can occasionally reproduce boxes of higher corresponding shared randomness dimension, effectively simulating $\abslambda>d$. See \citet[Sec. 4.1]{ZhangQuantumGameTheory} for a thorough discussion of both the existence and extent of super-locality in the special case $\sbell{2}{1}{v}$. Interestingly, super-locality can occur even in the absence of entanglement \cite{ScaraniUnpublished}.

\section{Regaining convexity by adding shared randomness} 

The quantum elliptope, i.e. $\mathcal{Q}_{\infty}$, is convex. The convexity of the quantum elliptope, i.e. $\mathcal{Q}_{\infty}$, is well established; see for example Refs. \citep[Sec. 5C]{WernerWolfBellInequalities} and \citep[Sec. 5]{FritzCombinatorialLong}; proof may also be given in terms of properties of C*-algebras \cite{KirchbergConjecture,TsirelsonProblemInteraction}.  We have shown in \cref{lem:Q2notconvex}, however, that for finite local Hilbert space dimension $\mathcal{Q}_{d}$ is sometimes not convex. The nonconvexity of $\mathcal{Q}_{d}$ was also demonstrated by \citet{ConcavityPaperVertesi}.  There is no contradiction between the convexity of $\mathcal{Q}_{\infty}$ and the nonconvexity of $\mathcal{Q}_{d}$; rather, convexification of quantum boxes requires either classical shared randomness or, equivalently, comes at the expense of increasing the local Hilbert space dimension. 

The reason quantum boxes cannot be mixed without increasing Hilbert space dimension is because the measurements are local. The local nature of the measurement operators means that the composite POVM element associated with some global input and output $\hat{M}_{ab...|xy...}=\hat{A}_{a|x}\otimes\hat{B}_{b|y}\otimes\!...$ is necessarily a product operator. Mixtures of product operators, sometimes known as separable superoperators~\cite{SeparableMeasurementsNathanson,SeparableMeasurementsRusso}, are generally no longer product operators.

On the other hand, access to classical randomness allows for the quantum preparations and measurements to co-depend on a shared classical hidden variable. Indeed, any combination of $N$ qudit-based boxes can be implemented by sharing a single qudit, so long as the qudit is prepared according to a classical variable $\lambda$ of dimension $N$, and the variable $\lambda$ remains accessible to the measurements, i.e. per \fig{fig:lqDAG}.  Explicitly, we imagine the quantum preparations and measurements to co-depend deterministically on $\lambda$, i.e. $p[\rho^{(\gamma)}|\lambda]=p[\hat{A}^{(\gamma)}_{a|x}|\lambda]=...=\delta[\gamma-\lambda]$. By this construction, the resulting hybrid quantum-classical boxes yields
\begin{align}\label{eq:generichybridbox}
p[ab...|xy...]=\sum_{\lambda=0}^{\abslambda-1}p[\lambda]\Tr{\rho^{(\lambda)} \, \hat{A}^{(\lambda)}_{a|x} \otimes \hat{B}^{(\lambda)}_{b|y} \,...}.
\end{align}  

Let's denote the set of boxes achievable using quantum systems of dimension $d$ with the assistance of shared randomness of dimension $\abslambda$ as $\mathcal{Q}_d+\mathcal{L}_{\abslambda}$.  A hybrid box can be converted into a purely quantum description by embedding the shared randomness into a block-diagonal classical-quantum (cq) state of dimension $d\times\abslambda$~\cite{CQstate2008,DimensionWitness2008Wehner,BaezDirectSum,BookDirectSum}, and as such:  \begin{axiom}\label{theo:LQinQ}
${\mathcal{L}_{\abslambda}+\mathcal{Q}_{d}}\subseteq\mathcal{Q}_{d^{\prime}\geq d\times \abslambda},\;$ i.e. any hybrid quantum-classic box can be thought of as a purely-quantum box predicated on a sufficiently larger local Hilbert space dimension $d^{\prime}= d\times \abslambda$.
\end{axiom}
\begin{proof}
To quote \citet{BaezDirectSum}: ``Roughly speaking, if we have a physical system whose states are either states of ${\rho}^{(i)}$ OR states of ${\rho}^{(j)}$, its Hilbert space will be the direct sum ${\mathcal{H}}^{(i)}\oplus{\mathcal{H}}^{(j)}$." The direct sum of $N$ Hilbert spaces each of dimension $d$ can be embedded in a single Hilbert space of dimension $d\times N$ \cite{BookDirectSum}. See \cref{sec:convexification} for further details.
\end{proof}

Having shown that $\mathcal{Q}_d$ can be non-convex and having established how to include finite shared randomness into our framework, we next ask how much supplementary shared randomness is required so that $\mathcal{L}_{\abslambda}+\mathcal{Q}_{d}$ will be convex. Equivalently, what is the worst-case number of different $P\in\mathcal{Q}_{d}$ that must be combined in order to simulate \emph{any} box which is expressible as a mixture of boxes from  $\mathcal{Q}_{d}$?  We rephrase that question into a definition.

\begin{defin}[\QCN]
The Quantum Carathéodory Number $\QCN^{\sbell{n}{m}{v}}$ is the minimum amount of supplementary classical correlation that must be supplied for $\mathcal{Q}^{\sbell{n}{m}{v}}_{d}$ to become completely convexified. Formally,
\begin{align}\begin{split}\label{eq:QCNnew}
\QCN &\equiv \min\; \abslambda\; \text{ such that }\;
\\&
\mathcal{L}_{\abslambda}+\mathcal{Q}^{\sbell{n}{m}{v}}_{d}=\CH{\mathcal{Q}_d}.
\end{split}\end{align}
We call $\QCN$ the quantum Carathéodory number because \cref{eq:QCNnew} dictates that every $P\in\CH{\mathcal{Q}_d}$ can be decomposed as a convex mixture of at-most $\QCN$ boxes from $\mathcal{Q}_d$, whereas ${\QCN-1}$ would not be adequate. 
\end{defin}

Although $\QCN$ may depend on $d$, there is a way to upper-bound the quantum Carathéodory number independently of $d$. 
\begin{theo}\label{theo:FenchelCountnew}
$\QCN \leq {\left(m(v-1)+1\right)}^n-1\,$.
\end{theo}
\begin{proof}
A 1929 theorem of Werner Fenchel \cite{Hanner1951,Brny2012} states that any point within the convex hull of a not-necessarily-convex $\mathcal{F}$-dimensional closed and pathwise-connected set can be decomposed as a convex mixture of at-most $\mathcal{F}$ points in the set; Fenchel's theorem is a strengthening of Carathéodory's theorem \cite{CaratheodoryMathworld} for the special case of indivisible sets such as pathwise-connected ones. $\mathcal{Q}_d$ is a pathwise-connected closed compactum of dimension $\mathcal{F}$; pathwise-connectedness follows from the continuous parameterization of both states and measurements and the Intermediate Value Theorem. 
The statistical space dimension of no-signalling boxes in a symmetric Bell scenarios is given by \cref{eq:dimcount}~\cite{PironioDimension,Brunner2013Bell}.
\end{proof}
For example, any box in $\CH{\mathcal{Q}^{\sbell{2}{2}{2}}_{d}}$ can be expressed as the convex combination of at-most eight qudit-based boxes. The actual quantum Carathéodory number may potentially be much lower than this upper bound, and is likely to have an explicit dependence on $d$. All we've established from \cref{lem:Q2notconvex} is that $\abs{\lambda^{\star}_{\mathcal{Q}}}_2\neq 1$. A tighter upper bound for $\QCN$ than is given in \cref{theo:FenchelCountnew} is desideratum for future research.

Importantly, the finiteness of ${|\lambda^{\star}_{\mathcal{Q}}|}{}_d$ can be used to guarantee the convexity of certain quantum sets ${\mathcal{Q}^{\sbell{n}{m}{v}}_{d}}$ for finite dimension $d$. 
\begin{theo}\label{theo:NoNewExtremePointsThenConvex}
If ${\mathcal{Q}_{\infty}=\CH{\mathcal{Q}_{d}}}$ for a given Bell scenario , i.e. all extremal quantum distributions are achievable with qudits of dimension $d$, then $\mathcal{Q}_{d'}$ is convex if $d'\geq{ d\times \QCN}$.
\end{theo}
\begin{cor}\label{cor:MasanesConvexity}
For \sbell{n}{2}{2} scenarios, $\mathcal{Q}_{d'}$ is convex if $d' \geq 2\left(3^n-1\right)$, or assuming the upper bound of \cref{theo:FenchelCountnew} is not tight, convexity is certain whenever $d' \geq 2 \abs{\lambda^{\star}_{\mathcal{Q}}}_2$.
\end{cor}
\begin{proof}
By the definition of the quantum Carathéodory number, $\CH{\mathcal{Q}_{d}}={\mathcal{L}_{\QCN}+\mathcal{Q}_d}$. By \cref{theo:LQinQ} we can embed the supplementary classical correlations into the shared quantum state such that $\CH{\mathcal{Q}_{d}}\subseteq\mathcal{Q}_{d'}$ whenever $d'\geq d\times \QCN$. However, the promise of ${\mathcal{Q}_{\infty}=\CH{\mathcal{Q}_{d}}}$ means that $\mathcal{Q}_{\infty}\subseteq\mathcal{Q}_{d'}$, from which it follows that $\mathcal{Q}_{\infty}=\mathcal{Q}_{d'}$. Thus, $\mathcal{Q}_{d'}$ is guaranteed to be convex. 

The corollary is a consequence of Masanes' theorem \cite{MasanesQubitsEarly,MasanesQubits}, which states that (projective) measurements on merely shared qubits are capable of achieving all extremal quantum distributions for scenarios involving two binary measurements per party \cite{FritzDuality}. Masanes' theorem is often cited when noting that the maximum violation of Bell inequalities for such scenarios can be computed by maximizing over qubit-based boxes.
\end{proof}

It is not clear if \cref{cor:MasanesConvexity} can be extended to more general scenarios.  For example, the $I_{3322}$ Bell inequality in the \sbell{2}{m}{2} many-measurement-choices scenario is apparently ever-more-violated as $d$ is increased \cite{I3322Original,I3322NPA2,InfiniteDimensionViolation,DimensionDependantBellViolation3}.  The boundaries of the quantum set increasing with dimension prevent using the above arguments, and indeed it remains an open question whether $\mathcal{Q}^{\sbell{2}{3}{2}}_d$ is convex for \emph{any} finite $d$.  On the other hand, there is numerical evidence that \sbell{2}{2}{v} many-measurement-outcomes scenarios, such as the CGLMP scenario~\cite{CGLMP02}, achieve maximum nonlocality at $d=v$~\cite{CGLMP08,NPA2008Long}. The maximum violation of every Bell inequality, however, does not necessarily imply that all quantum extremal distributions have been achieved, so it is not certain that convexity is achieved. See Ref. \citep[Sec. III.B]{Brunner2013Bell} for further details.

\begin{conj}\label{conj:artifact}
We conjecture that nonconvexity is nothing more than an artifact of not spanning the local polytope, i.e. that if $\mathcal{L}_{\abslambdastar}\subseteq{\mathcal{L}_{\abslambda}+\mathcal{Q}_d}$ then ${\mathcal{L}_{\abslambda}+\mathcal{Q}_d}$ is convex. %
\end{conj}
In stating \cref{conj:artifact} we used the equivalence ${\mathcal{L}_{\infty}=\mathcal{L}_{\abslambdastar}}$ per \cref{eq:minsharedrandom}. Recall that the converse of \cref{conj:artifact} is obviously true: if a set of boxes -- quantum, classical, or hybrid -- does not contain the local polytope, then the set is not convex per \cref{cor:SubLocalNonConvex}.

\cref{conj:artifact} amounts to speculating that $\QCN=\min\; \abslambda$ such that $\mathcal{L}_{\abslambdastar}\subseteq{\mathcal{L}_{\abslambda}+\mathcal{Q}_d}$. If true, this would replace \cref{theo:FenchelCountnew} with the claim  $\QCN \leq \ceil{\abslambdastar/d}.$ This follows from $\mathcal{L}_{d}\subseteq\mathcal{Q}_{d}$ per \cref{theo:LinsideQ} and then by  $\mathcal{L}_{\abslambdastar}\subseteq{\mathcal{L}_{\ceil{{\abslambdastar}/{d}}}+\mathcal{L}_{d}}$, which merely notes that a random variable can be always be decomposed into multiple constituent parts. Equivalently, any integer $[1,N]$ can be mapped injectively to an ordered tuple $\brackets{[1,M],[1,\ceil{{N}/{M}}]}$.

To be clear, it is an open question if $\mathcal{Q}_d$ is \emph{ever} convex for finite $d$ in scenarios where \cref{theo:NoNewExtremePointsThenConvex} does not apply. \cref{conj:artifact} not only speculates the affirmative but effectively proposes a minimal value for what that finite $d$ might be.  Finally, note that an implication of \cref{conj:LnotsubsumedbysmallQ} and \cref{conj:artifact} combined is that $\mathcal{Q}_d$ should be convex if and only if $d\geq\abslambdastar$. 

\section{Discussion}

Questions remain even when still considering the \sbell{2}{2}{2} scenario.  We've established that $\mathcal{L}^{\sbell{2}{2}{2}}_{\infty}\nsubseteq\mathcal{Q}_{2}$, but would qutrits be able to span the local polytope, or would $d=4$ be required?\footnote{
\(\abslambdastar^{\sbell{2}{2}{2}}=4\) is proven in the \cref{sec:lambdastar}. \label{foot:L4sufficient}} The insufficiency of qutrits is speculated in \cref{conj:LnotsubsumedbysmallQ}, but this should certainly be investigated further.  Furthermore, although from \cref{cor:MasanesConvexity} is is clear that $\mathcal{Q}^{\sbell{2}{2}{2}}_{16}$ is convex, there's still a large gap between the non-convex result of $\mathcal{Q}_{2}$ and the yes-convex result of $\mathcal{Q}_{16}$. The convexity of $\mathcal{Q}^{\sbell{2}{2}{2}}_{4}$ is speculated in \cref{conj:artifact}, but this too should certainly be investigated further. For general scenarios where one presumes that \cref{theo:NoNewExtremePointsThenConvex} does not apply, such as \sbell{2}{3}{2} \cite{I3322NPA2}, we have noted it is a completely open question if $\mathcal{Q}_{d}$ is \emph{ever} convex for finite $d$.

Most generally, given three descriptive elements: 1) an operational description of some Bell scenario such as \sbell{n}{m}{v}, 2) the local Hilbert space dimension limit of $d$, and 3) the dimension limit of any shared classical randomness $\abslambda$, are the resulting correlations ${\mathcal{L}_{\abslambda}+\mathcal{Q}_d}$ convex?  Such fundamental questions remain generally unanswered despite the broad results of \cref{theo:FenchelCountnew,theo:NoNewExtremePointsThenConvex}. The purely classical regime is considered further, however, in the \cref{sec:lambdastar}. 

Quantifying the genuine boundaries of finite-dimensionally-generated quantum correlations, as opposed to the convex hull of such correlations, is an important question for future research.  We have evidenced that nonconvexity should be expected in all nonlocality scenarios with sufficiently restrictive constraints on the local Hilbert space dimensions.  As many physical systems have implicit bounds on their local Hilbert space dimensions it is all the more important to anticipate nonconvexity of quantum correlations in practical quantum information-theoretic protocols.  Purely quantum systems with sufficiently low dimensionality are forbidden from displaying certain classical correlations, a property which may be exploited as a device-authenticating security check in quantum cryptographic implementations.  

The nonconvexity of quantum correlations is no less relevant, practically, than other no-go results pertaining to finite Hilbert space dimensions. Quantifying the \emph{genuine} boundaries of finite-dimensionally-generated quantum correlations, as opposed to the convex hull of such correlations \cite{Navascues2013FiniteDimensions,Navascues2014FiniteDimensions}, is therefore an important question for future research.  How the nonconvexity of dimensionally limited correlations may clarify the relationship between Hilbert space dimension and degree of classical randomness is an intriguing area for future work~\cite{RandomnessVNonlocality,NonlocalityAnomaly,DimensionAsResource}.  Perhaps classical information theory tools for considering limited shared randomness \cite{GeometryAndNovel} can be adapted and applied to finite dimensional quantum systems.  Fuller quantitative operational characterizations of finite dimension quantum systems are thus valuable desiderata for both foundational and practical research.

\begin{acknowledgments}
We are grateful to Joshua Combes, Tobias Fritz, Ravi Kunjwal, Corsin Pfister, Marco Quintino, Sandu Popescu, Matt Pusey, Kevin Resch, Vincent Russo, and Rob Spekkens for valuable discussions, and to Matty Hoban for alerting us to a serious flaw in an early version of this manuscript. J.~M.~D. is grateful for support from the Natural Sciences and Engineering Research Council of Canada. Research at Perimeter Institute is supported by the Government of Canada through Industry Canada and by the Province of Ontario through the Ministry of Economic Development and Innovation.
\end{acknowledgments}

\appendix

\section{Correlations from asymmetric-dimensional quantum states}\label{sec:asymmdim}

In \cref{sec:notation}, we defined $\mathcal{Q}_d^{\sbell{n}{m}{v}}$ as the set of correlations achievable when all parties share a quantum state with \emph{equal} local dimensions $d$.  Here we note that any result concerning the uniform-dimension case also applies to the asymmetric-dimension case if the scenario in question is bipartite.  We are grateful to referee comments for suggesting this extension.

\begin{prop}\label{theo:asymmbipartite}
If, in some bipartite Bell scenario, Alice and Bob share a quantum system with local dimensions $d_A$ and $d_B$, respectively, the set of distributions they can achieve is identical to the set of distributions achievable if the minimum dimension were common across both parties:  $\mathcal{Q}_{\{d_A,d_B\}}^{\sbell{2}{m}{v}} = \mathcal{Q}_{d_{min}}^{\sbell{2}{m}{v}}$, where $d_{min}=\min \{d_A,d_B\}$.
\end{prop}
\begin{proof}
In realizing achievable quantum distributions in a specified dimension, it suffices to consider pure states as the mixedness may be included in the local measurement settings; see \citep[Lemma 1]{ConditionalMinimumHSD}.  Any pure state in ${{\mathcal{H}}^{d_A}\otimes{\mathcal{H}}^{d_B}}$ is equivalent under local unitary transformations to its Schmidt decomposition, which has at most ${d_{min}=\min\{d_A,d_B\}}$ terms~\cite{acin2000generalized,characterizingentanglement,multireview,entang.review.toth}.  Therefore, any correlation achievable with states in ${{\mathcal{H}}^{d_A}\otimes{\mathcal{H}}^{d_B}}$ may also be achieved with states in ${{\mathcal{H}}^{d_{min}}\otimes{\mathcal{H}}^{d_{min}}}$.
\end{proof}

Note that this proof relies on the fact that the Schmidt decomposition of a bipartite state can be written as a sum of at most $d_{min}$ terms.  This feature does not extend to general entanglement classes in multipartite scenarios~\cite{acin2000generalized}; in these scenarios, asymmetric-dimensional states may lead to unique sets of achievable distributions.

\section{Bounding the Minimum Classical Shared Randomness which Spans the Local Polytope}\label{sec:lambdastar}

In this section we establish some upper and lower bounds on $\abslambdastar$, as defined per \ceq{eq:minsharedrandom} in the main text. We are grateful to Matt Pusey of the Perimeter Institute for Theoretical Physics for suggesting most of these proofs. The related question of how to find an explicit classical implementation of a box while minimizing $\abslambda$ (``ontological compression") is considered in Ref. \cite{ClassicalLambdaRequirement}.

\begin{prop}\label{theo:minsharedrandomnesssufficientPusey}
$\displaystyle\abslambdastar\leq v^{m(n-1)}\,$. In particular for the \sbell{2}{2}{2} scenario \(\displaystyle\abslambdastar\leq 4\,.\)
\end{prop}
\begin{proof}Let $n-1$ parties have all their measurements depend deterministically on $\lambda$, and only the $n$'th party depending probabilistically on $\lambda$. There are precisely $v^{m(n-1)}$ possible deterministic distributions among the $n-1$ parties, so without loss of generality is suffices to take $\abslambda = v^{m(n-1)}\,$. In this manner every conceivable classical correlation is possible.
\end{proof}

\begin{prop}\label{theo:minsharedrandomnesssufficientFenchel}
$\displaystyle\abslambdastar \leq {\left(m(v-1)+1\right)}^{n}-1\,$.
\end{prop}
\begin{proof}
The proof is analogous to the proof of \cref{theo:FenchelCountnew}. $\mathcal{L}_1$ is a pathwise-connected closed compactum of dimension $\mathcal{F}^{\sbell{n}{m}{v}}={\left(m(v-1)+1\right)}^{n}-1$ per \cref{eq:dimcount}. Fenchel's theorem \cite{Hanner1951,Brny2012} then implies that the Carathéodory number of any box $\mathcal{L}_1$ is less than or equal to $\mathcal{F}$; i.e. any $P\in\CH{\mathcal{L}_{1}}$ can be decomposed into a convex combination of at most $\mathcal{F}$ boxes from $\mathcal{L}_1$. $\abslambdastar$ is identically that Carathéodory of  $\mathcal{L}_1$.
\end{proof}


\begin{prop}\label{theo:maxsharedrandomnesssufficientXOR}
\(\displaystyle\abslambdastar\geq {\left(m(v-1)+1\right)}^{n-1}+\sum_{k=2}^n\sum_{j=2}^k \binom{n-1}{k-1}\binom{m}{j}S_2[k,j](j-1)!{(v-1)}^k\,,\) where $S_2$ refers to Stirling number of the second kind \cite{StirlingS2Mathworld}. For $n=2$, therefore, \(\displaystyle\,\abslambdastar\geq {\frac{m(m-1)(v-1)^2}{2}+m(v-1)+1}\,\) and in particular for the \sbell{2}{2}{2} scenario \(\displaystyle\abslambdastar\geq 4\,.\)
\end{prop}
\begin{proof}
Consider a box which is as random as possible while still satisfying \begin{align}
p[{a\mathpunct{+}b\mathpunct{+}c...=0 \!\!\!\mathopen{\mod v}}\mid xyz]=1 \quad \text{iff }\quad x=y=z...
\end{align}
i.e. the outputs of the parties always (modularly) sum to zero when the inputs are aligned, but this perfect correlation is not detected whenever the inputs are not all aligned. One can then, given $n-1$ of the outputs, perfectly determine the remaining one as $a_x=-b_x-c_x... \!\!\!\mathopen{\mod v}$. The instances of perfect correlation enforce that there is no local noise, and thus that every party's output depends \emph{deterministically} on the random variable.  The degrees of freedom in the outputs of the first $n-1$ parties, given by \ceq{eq:dimcount} as ${\left(m(v-1)+1\right)}^{n-1}-1$, are all able to be set independently and yet must be determined solely by the shared random variable. We must give the shared variable an alphabet size equal to the number of degrees of freedom \emph{plus one}, to account for normalization.  Thus, $\abslambdastar\geq{\left(m(v-1)+1\right)}^{n-1}$.

This loose lower bound can be strengthened to the expression in \cref{theo:maxsharedrandomnesssufficientXOR} by counting the remaining degrees of freedom which include the last party's outcome.  There are $\binom{n-1}{k-1}$ ways to choose a $k$-partite context involving the last party. We then need to assign inputs to the parties; as the case where all inputs are equal is already specified by the definition of the box, we only need to consider distributing $j>1$ distinct inputs among the $k$ parties. There are $\binom{m}{j}$ ways to choose which $j$ distinct inputs to distribute. The Stirling number $S_2[k,j]$ gives the number of ways $k$ objects can be divided into $j$ (indistinguishable) partitions. Each partition is assigned one input however, so we must consider the various permutations of mapping inputs to partitions. To avoid over-specifying the distribution we take the last party (and any other parties in the partition including the last party) to be associated with largest of the $j$ selected inputs. Thus we only consider $(j-1)!$ permutations of how the smaller inputs can be assigned to the remaining partitions. Finally we specify the outputs, avoiding output $0$ per the parameterization scheme discussed prior to \cref{eq:dimcount}. Thus there are $v-1$ distinct outputs considered for each of the $k$ parties considered.
\end{proof}

For pedagogical clarity we demonstrate how to obtain \(\abslambdastar\geq 7\) in this fashion for the \sbell{2}{3}{2} scenario, famous for the Bell inequality $I_{3322}$ \cite{I3322NPA2,InfiniteDimensionViolation,DimensionDependantBellViolation3}. The classical box in \sbell{2}{3}{2} which satisfies $b_x\mathopen{=}a_x$ has six degrees of freedom, namely ${p[a_0\mathopen{=}1]}$, ${p[a_1\mathopen{=}1]}$, ${p[a_2\mathopen{=}1]}$, ${p[a_1\mathopen{=}1,b_1\mathopen{=}1]}$, ${p[a_0\mathopen{=}1,b_2\mathopen{=}1]}$, and ${p[a_1\mathopen{=}1,b_2\mathopen{=}1]}$. Trivially ${p[a_x\mathopen{=}0]}=1-{p[a_x\mathopen{=}1]}$ and ${p[a_x\mathopen{=}1,b_y\mathopen{=}0]}={p[a_x\mathopen{=}1]}-{p[a_x\mathopen{=}1,b_y\mathopen{=}1]}$ etc. Note that $p[b_y\mathopen{=}i]$ is fixed as equal to $p[a_y\mathopen{=}i]$. Furthermore, $p[a_x\mathopen{=}i,b_y\mathopen{=}j]=p[a_y\mathopen{=}j,b_x\mathopen{=}i]$, so we only ever enumerate $p[a_x\mathopen{=}i,b_y\mathopen{=}j]$ where $x<y$. Adding one to the six degrees of freedom yields \(\abslambdastar\geq 7\).


\section{A Complimentary Fundamental Axiom}\label{sec:compaxiom}

In \cref{theo:LinsideQ}, it was noted that the set of boxes achievable with classical shared randomness of dimension $\abslambda$ is achievable with quantum states of dimension $d=\abslambda$, and furthermore, realizable with separable quantum states.  Here, we briefly note that while the inverse is not true, it \emph{is} true that any box achievable with separable quantum states of dimension $d$ can be achieved with classical shared randomness of finite dimension.

\begin{axiom}\label{theo:QinsideL}
${{(\mathcal{Q}:\mbox{sep})}_{\mathlarger d}\subseteq\mathcal{L}_{\abslambda\geq \mathlarger{d^n}}},$ i.e. all quantum correlations generated by separable states are also classically achievable given enough shared randomness, where enough means $\abslambda\geq \mathlarger{d^n}$.
\end{axiom}
\begin{cor}
Entanglement is required for non-locality.
\end{cor}
\begin{proof}
Recall that by definition separable states can be written as $\rho_{\text{sep}}=\sum_{i}{c^{(i)} \rho_A^{(i)} \otimes  \rho_B^{(i)}...}\,$, and therefore satisfy \(\Tr{\rho_{\text{sep}}\,\left(\hat{A}_{a|x}\otimes\hat{B}_{b|y}\,...\right)}=
 \sum_{i=1}^{i_\text{max}}\left({c^{(i)} \Tr{\rho_A^{(i)} \, \hat{A}_{a|x} }\Tr{\rho_B^{(i)} \, \hat{B}_{b|y} }...}\right)\,
\). Any separable state with local Hilbert space dimension $d$ can be decomposed into a mixture of no-more-than $d^n$ product states [\citealp{NSepStateMixing} Def. 6, \citealp{SeparableEnsembleLength} Thm. 2] so we can replace $i_\text{max}$ with $d^n$ without loss of generality. One can therefore map every $P\in \mathcal{(\mathcal{Q}:\mbox{sep})}_{d}$ to a $P^{\prime}\in \mathcal{L}_{\abslambda= \mathlarger{d^n}}$ by the following construction on \ceq{eq:genericclassicalbox}: $\lambda\to i$, $p[\lambda]\to c^{(i)}$, $p[a|x\lambda]\to\Tr{\rho_A^{(i)} \, \hat{A}_{a|x}}$, etc.
\end{proof}

\citet{ScaraniUnpublished} have shown that merely $\abslambda\geq d$ is \emph{not} always sufficient to classically simulate the correlations which result from separable states, i.e. that separable states can still manifest super-locality. In particular, they show that the box $p[ab|xy]=\frac{2+\left(-1\right)^{a+b+{x y}}\sqrt{2}}{8}$ requires $\abslambda>2$, but can nevertheless be achieved by measurements on separable qubits.

\section{Details of the Parameterization of Shared-Qubits Boxes in the \sbell{2}{2}{2} Scenario}\label{sec:params}

We consider qubit-based boxes with explicit representations in terms of two-qubit states of arbitrary entanglement and general two-outcome POVMs. We need to consider exclusively pure states, per Ref. \citep[Lemma 1]{ConditionalMinimumHSD}.  As all pure states are equivalent to their Schmidt-decomposed form under local unitary transformations~\cite{acin2000generalized,characterizingentanglement,multireview,entang.review.toth}, it is thus sufficient to consider the state \begin{align}
\ket{\psi}=\cos(\tfrac{\alpha}{2})\ket{00}+\sin(\tfrac{\alpha}{2})\ket{11}\,,\quad\alpha\in(0,\pi),
\label{genstate}\end{align} by folding the local degrees of freedom into the measurement operators.

While POVMs can be regarded as projective measurements in a larger Hilbert space, restricting the dimensionality of the quantum systems necessitates the use of general POVMs~\cite{VBInequality,ConcavityPaperVertesi}.  Adapting the notation of~\cite{HeinosaariPOVM}, we express a general binary 0/1 outcome POVM element as \begin{align}\label{genmeas}
\hat{A}_{a|x}&=\frac{1}{2}\left[\left(1+(-1)^{1-a}\kappa_{A_x}\right)\,\id+(-1)^{1-a}\eta_{A_x}\left(\vec{n}_{A_x}\cdot\vec{\sigma}\right) \right],\end{align} where $\vec{\sigma}=(\hat{\sigma}_X,\hat{\sigma}_Y,\hat{\sigma}_Z)$ is a vector of Pauli matrices and $\vec{n}_{A_x}=(\sin\theta_{A_x}\cos\phi_{A_x},\sin\theta_{A_x}\sin\phi_{A_x},\cos\theta_{A_x})$ is a unit vector defining a direction in the Bloch sphere in spherical coordinates. Bob's POVM elements are defined similarly in a separate Hilbert space.  Technically, the measurement operators of the distinct parties need only commute with each other to form a genuine quantum multipartite implementation; we relegate each party to a distinct Hilbert space for convenience to ensure appropriate commutativity, which has apparently no loss of generality \cite{TsirelsonProblemArXiv,ConnesEmbedding,PhysicalTsirelson,KirchbergConjecture}.  To ensure positivity of the POVM elements corresponding to both outputs, the following conditions must be met:
\begin{align}
\forall_{A_x}:\eta_{A_x}-1\leq\kappa_{A_x}\leq1-\eta_{A_x}.\label{genmeasconds}
\end{align}
\noindent Note that \ceq{genmeasconds} implies $0\leq\eta_{k_x}\leq1$. If both bounds in \ceq{genmeasconds} are simultaneously saturated then $\kappa_{k_x}=0$ and $\eta_{k_x}=1$, and \ceq{genmeas} represents a projection-valued measurement (PVM).

As we are concerned with two-outcome POVMs, it is conventional to parameterize boxes in terms of bias of the measurement outcome, i.e. \begin{align}\Braket{A_x}=&\langle\hat{A}_{1|x}\rangle-\langle\hat{A}_{0|x}\rangle\\ \begin{split}\Braket{A_xB_y}=&\langle\hat{A}_{1|x}\hat{B}_{1|y}\rangle-\langle\hat{A}_{0|x}\hat{B}_{1|y}\rangle\\&- \langle\hat{A}_{1|x}\hat{B}_{0|y}\rangle+\langle\hat{A}_{0|x}\hat{B}_{0|y}\rangle.\end{split}\end{align}  The four marginal and four joint biases (for all $x$ and $y$ options) parameterize the eight-dimensional conditional probability space for the \sbell{2}{2}{2} scenario, and relate to the 17 ``backend'' parameters of the state and measurements as
\begin{align}\label{eq:POVMmarginal}
\Braket{A_x}\!=&\eta_{A_x}\cos(\alpha)\cos(\theta_{A_x})+\kappa_{A_x}\,\text{, and}
\\\nonumber
\Braket{A_x B_y}\!=&\eta_{A_x}\eta_{B_y}\cos(\phi_{A_x}\!+\phi_{B_y})\sin(\alpha)\sin(\theta_{A_x})\sin(\theta_{B_y})\\&+\eta_{A_x}\eta_{B_y}\cos(\theta_{A_x})\cos(\theta_{B_y})\nonumber\\&+\eta_{A_x}\kappa_{B_y}\cos(\alpha)\cos(\theta_{A_x})\label{eq:POVMparity}\\&+\eta_{B_y}\kappa_{A_x}\cos(\alpha)\cos(\theta_{B_y})+\kappa_{A_x}\kappa_{B_y}\,.\nonumber
\end{align}


\section{Convexly Combining Quantum Boxes via Direct Sum of Hilbert Spaces}\label{sec:convexification}

Suppose we wish to take the convex combination of $N$ qudit-based multipartite boxes, i.e.
\begin{align}\label{eq:convexcomboP}
\widetilde{P}=\sum_{i=0}^{N-1}{c_{i}P_{i}}\quad\text{where}\quad P_{i}\in\mathcal{Q}_{d}\,.
\end{align}
Typically\footnote{The major exception is when $\mathcal{Q}_{d}$ is convex, as then $\widetilde{P}\in\mathcal{Q}_{d}$ by definition. Convexity might be ``rare", however.\label{foot:Q16sufficient}}, $\widetilde{P}\not\in\mathcal{Q}_{d}\,$. This can be built as follows:

Index the local measurement operators of each of the $N$ boxes being combined into $\widetilde{P}$ by $\hat{A}^{(i)}_{a|x},\hat{B}^{(i)}_{b|y}$ etc. Now imagine -- although entirely unjustified -- that all the $N$ boxes $P_{i}$ are predicated on sharing the \emph{same} composite quantum state, i.e. $\forall_i\;\;\rho_i=\rho_0$. In order to reproduce the marginal probabilities of $\widetilde{P}$ with a single quantum box (without supplementary shared randomness), we should take $\widetilde{\rho}=\rho_0$ and define the new measurement operators $\widetilde{\hat{A}_{a|x}}=\sum_{i=0}^{N-1}c_i\,\hat{A}^{(i)}_{a|x}$. This satisfies the requirement that $\widetilde{p}[a|x]=\sum_{i=0}^{N-1}c_i\, {p_i}[a|x]$ etc.

Unfortunately, choosing measurement operators to satisfy the single-partite marginal probabilities does not extend to satisfaction of the required bipartite joint probabilities. Following \ceq{eq:genericquantumbox} we find that
\begin{align}\begin{split}
p[ab|xy]&=\Tr{\widetilde{\rho} \, \widetilde{\hat{A}_{a|x}} \otimes \widetilde{\hat{B}_{b|y}}}\\&
=\sum_{i=0}^{N-1}\sum_{j=0}^{N-1}c_i\, c_j \Tr{\rho_0 \, \hat{A}^{(i)}_{a|x} \otimes \hat{B}^{(j)}_{b|y}}
\end{split}\end{align}
which does not remotely match
$\widetilde{p}[ab|xy]
{=\sum_{i}^{}{c_i\,}{p_i}[ab|xy]}
{=
\sum_{i}^{}{c_i}\Tr{\rho_0 \, \hat{A}^{(i)}_{a|x} \otimes \hat{B}^{(i)}_{b|y}}}$.  Therefore quantum boxes predicated on different sets of local measurement operators cannot be combined without increasing the Hilbert space dimension\footnote{Suppose we artificially consider global measurements of the form $\widehat{A'' B''}=\frac{\hat{A}\otimes\hat{B}+\hat{A'}\otimes\hat{B'}}{2}$. Such not-local-but-still-separable measurements lack physical meaning, but they can be interpreted as representing the convex combination of boxes acting on the same shared state with different local, i.e. product state, measurements. Such artificial global measurements are known as \emph{separable superoperators}, and have been studied elsewhere in the context of state discrimination \cite{SeparableMeasurementsNathanson,SeparableMeasurementsRusso}.}; as per \cref{theo:LQinQ}, however, $\widetilde{P}\in\mathcal{Q}_{d\times N}$.

Imagine that each $P_i$ is implemented in ${\mathcal{H}}^{(i)}={\left(\mathbb{C}^{d}\right)}^{\otimes n}$ using the (distinct!) quantum states $\rho^{(i)}$ and local measurement operators $\hat{A}^{(i)}_x$, $\hat{B}^{(i)}_y$, etc. Then, to implement $\widetilde{P}$ using a single quantum box, one may take the direct sum of the component Hilbert spaces such that
\begin{align}\label{eq:directsumH}
\widetilde{\mathcal{H}}=\bigoplus\limits_{i=0}^{N-1}{\mathcal{H}}^{(i)}\,,\quad \widetilde{\rho}=\bigoplus\limits_{i=0}^{N-1}{c_i\,}{\rho_i}\,,\quad\widetilde{\hat{A}_{a|x}}=\bigoplus\limits_{i=0}^{N-1}\hat{A}^{(i)}_{a|x},
\end{align}
and so on.

To explain how \ceq{eq:directsumH} automatically satisfies \ceq{eq:convexcomboP}, and why the new local Hilbert space dimension in \ceq{eq:directsumH} can be thought of as $d\times N$, consider how $\widetilde{P}$ is implemented in ${\big({\left(\mathbb{C}^{d}\right)}^{\otimes N}\big)}^{\otimes n}$. The idea is to assign an ancilla $\mathbb{C}^{N}$ to every party, and to make the new shared state simultaneously diagonalized in all the ancillae. Thus,
\begin{align}\begin{split}\label{eq:highdimconstruction}
\widetilde{\rho}&=\sum\limits_{i=0}^{N-1}{{c_i\,}\rho_i\otimes \dyad{i}{i}_A \otimes \dyad{i}{i}_B}\otimes...\,,
\\ \widetilde{\hat{A}_{a|x}} &= \sum\limits_{i=0}^{N-1}{\hat{A}^{(i)}_{a|x}\otimes \dyad{i}{i}_A },\\\widetilde{\hat{B}_{b|y}} &= \sum\limits_{j=0}^{N-1}{\hat{B}^{(j)}_{b|y} \otimes \dyad{j}{j}_B }\,,
\end{split}\end{align}
and so on. Indeed, \ceq{eq:highdimconstruction} amounts to the \emph{definition} of the direct sum in \ceq{eq:directsumH}. Thus the parties' new measurement operators now live in distinct Hilbert spaces $N$ times larger than originally, and the new shared state is a cq-state \cite{CQstate2008,DimensionWitness2008Wehner}, namely a block-diagonal composition of the component original states in the convex combination per \ceq{eq:convexcomboP}.

What we find is that \emph{direct summation of Hilbert spaces is identically convexification} in the sense of \ceq{eq:generichybridbox}. Since ${\bigoplus^N_{}\bigotimes^n_{}\mathbb{C}^{d} \subset \bigotimes^N_{}\bigotimes^n_{}\mathbb{C}^{d}}$, we have proven that  ${\mathcal{L}_{N}+\mathcal{Q}_{d}\subseteq\mathcal{Q}_{d\times N}}$.

\clearpage

\setlength{\bibsep}{1pt plus 30pt minus 2pt}
\bibliographystyle{apsrev4-1}
\nocite{apsrev41Control}
\bibliography{apsrevcontrol,nonconvexityrefs}

\begin{thebibliography}{86}%
\makeatletter
\providecommand \@ifxundefined [1]{%
 \@ifx{#1\undefined}
}%
\providecommand \@ifnum [1]{%
 \ifnum #1\expandafter \@firstoftwo
 \else \expandafter \@secondoftwo
 \fi
}%
\providecommand \@ifx [1]{%
 \ifx #1\expandafter \@firstoftwo
 \else \expandafter \@secondoftwo
 \fi
}%
\providecommand \natexlab [1]{#1}%
\providecommand \enquote  [1]{``#1''}%
\providecommand \bibnamefont  [1]{#1}%
\providecommand \bibfnamefont [1]{#1}%
\providecommand \citenamefont [1]{#1}%
\providecommand \href@noop [0]{\@secondoftwo}%
\providecommand \href [0]{\begingroup \@sanitize@url \@href}%
\providecommand \@href[1]{\@@startlink{#1}\@@href}%
\providecommand \@@href[1]{\endgroup#1\@@endlink}%
\providecommand \@sanitize@url [0]{\catcode `\\12\catcode `\$12\catcode
  `\&12\catcode `\#12\catcode `\^12\catcode `\_12\catcode `\%12\relax}%
\providecommand \@@startlink[1]{}%
\providecommand \@@endlink[0]{}%
\providecommand \url  [0]{\begingroup\@sanitize@url \@url }%
\providecommand \@url [1]{\endgroup\@href {#1}{\urlprefix }}%
\providecommand \urlprefix  [0]{URL }%
\providecommand \Eprint [0]{\href }%
\providecommand \doibase [0]{http://dx.doi.org/}%
\providecommand \selectlanguage [0]{\@gobble}%
\providecommand \bibinfo  [0]{\@secondoftwo}%
\providecommand \bibfield  [0]{\@secondoftwo}%
\providecommand \translation [1]{[#1]}%
\providecommand \BibitemOpen [0]{}%
\providecommand \bibitemStop [0]{}%
\providecommand \bibitemNoStop [0]{.\EOS\space}%
\providecommand \EOS [0]{\spacefactor3000\relax}%
\providecommand \BibitemShut  [1]{\csname bibitem#1\endcsname}%
\let\auto@bib@innerbib\@empty
\bibitem [{\citenamefont {Bell}(1964)}]{bell1964einstein}%
  \BibitemOpen
  \bibfield  {author} {\bibinfo {author} {\bibfnamefont {J.~S.}\ \bibnamefont
  {Bell}},\ }\bibfield  {title} {\enquote {\bibinfo {title} {On the
  {E}instein-{P}odolsky-{R}osen paradox},}\ }\href
  {http://cds.cern.ch/record/111654/files/} {\bibfield  {journal} {\bibinfo
  {journal} {Physics}\ }\textbf {\bibinfo {volume} {1}},\ \bibinfo {pages}
  {195} (\bibinfo {year} {1964})}\BibitemShut {NoStop}%
\bibitem [{\citenamefont {Brunner}\ \emph {et~al.}(2014)\citenamefont
  {Brunner}, \citenamefont {Cavalcanti}, \citenamefont {Pironio}, \citenamefont
  {Scarani},\ and\ \citenamefont {Wehner}}]{Brunner2013Bell}%
  \BibitemOpen
  \bibfield  {author} {\bibinfo {author} {\bibfnamefont {N.}~\bibnamefont
  {Brunner}}, \bibinfo {author} {\bibfnamefont {D.}~\bibnamefont {Cavalcanti}},
  \bibinfo {author} {\bibfnamefont {S.}~\bibnamefont {Pironio}}, \bibinfo
  {author} {\bibfnamefont {V.}~\bibnamefont {Scarani}}, \ and\ \bibinfo
  {author} {\bibfnamefont {S.}~\bibnamefont {Wehner}},\ }\bibfield  {title}
  {\enquote {\bibinfo {title} {{Bell nonlocality}},}\ }\href {\doibase
  10.1103/RevModPhys.86.419} {\bibfield  {journal} {\bibinfo  {journal} {Rev.
  Mod. Phys.}\ }\textbf {\bibinfo {volume} {86}},\ \bibinfo {pages} {419}
  (\bibinfo {year} {2014})}\BibitemShut {NoStop}%
\bibitem [{\citenamefont {Popescu}(2014)}]{PopescuReviewNatureComm}%
  \BibitemOpen
  \bibfield  {author} {\bibinfo {author} {\bibfnamefont {S.}~\bibnamefont
  {Popescu}},\ }\bibfield  {title} {\enquote {\bibinfo {title} {{Nonlocality
  beyond quantum mechanics}},}\ }\href {\doibase 10.1038/nphys2916} {\bibfield
  {journal} {\bibinfo  {journal} {Nat. Phys.}\ }\textbf {\bibinfo {volume}
  {10}},\ \bibinfo {pages} {264} (\bibinfo {year} {2014})}\BibitemShut
  {NoStop}%
\bibitem [{\citenamefont {Barrett}\ \emph {et~al.}(2005)\citenamefont
  {Barrett}, \citenamefont {Linden}, \citenamefont {Massar}, \citenamefont
  {Pironio}, \citenamefont {Popescu},\ and\ \citenamefont
  {Roberts}}]{NoSigPolytope}%
  \BibitemOpen
  \bibfield  {author} {\bibinfo {author} {\bibfnamefont {J.}~\bibnamefont
  {Barrett}}, \bibinfo {author} {\bibfnamefont {N.}~\bibnamefont {Linden}},
  \bibinfo {author} {\bibfnamefont {S.}~\bibnamefont {Massar}}, \bibinfo
  {author} {\bibfnamefont {S.}~\bibnamefont {Pironio}}, \bibinfo {author}
  {\bibfnamefont {S.}~\bibnamefont {Popescu}}, \ and\ \bibinfo {author}
  {\bibfnamefont {D.}~\bibnamefont {Roberts}},\ }\bibfield  {title} {\enquote
  {\bibinfo {title} {Nonlocal correlations as an information-theoretic
  resource},}\ }\href {\doibase 10.1103/PhysRevA.71.022101} {\bibfield
  {journal} {\bibinfo  {journal} {Phys. Rev. A}\ }\textbf {\bibinfo {volume}
  {71}},\ \bibinfo {pages} {022101} (\bibinfo {year} {2005})}\BibitemShut
  {NoStop}%
\bibitem [{\citenamefont {Barrett}\ and\ \citenamefont
  {Pironio}(2005)}]{PRUnit}%
  \BibitemOpen
  \bibfield  {author} {\bibinfo {author} {\bibfnamefont {J.}~\bibnamefont
  {Barrett}}\ and\ \bibinfo {author} {\bibfnamefont {S.}~\bibnamefont
  {Pironio}},\ }\bibfield  {title} {\enquote {\bibinfo {title}
  {{{P}opescu-{R}ohrlich correlations as a Unit of nonlocality}},}\ }\href
  {\doibase 10.1103/PhysRevLett.95.140401} {\bibfield  {journal} {\bibinfo
  {journal} {Phys. Rev. Lett.}\ }\textbf {\bibinfo {volume} {95}},\ \bibinfo
  {pages} {140401} (\bibinfo {year} {2005})}\BibitemShut {NoStop}%
\bibitem [{\citenamefont {Jones}\ and\ \citenamefont
  {Masanes}(2005)}]{NoSigPolytope2}%
  \BibitemOpen
  \bibfield  {author} {\bibinfo {author} {\bibfnamefont {N.~S.}\ \bibnamefont
  {Jones}}\ and\ \bibinfo {author} {\bibfnamefont {L.}~\bibnamefont
  {Masanes}},\ }\bibfield  {title} {\enquote {\bibinfo {title} {Interconversion
  of nonlocal correlations},}\ }\href {\doibase 10.1103/PhysRevA.72.052312}
  {\bibfield  {journal} {\bibinfo  {journal} {Phys. Rev. A}\ }\textbf {\bibinfo
  {volume} {72}},\ \bibinfo {pages} {052312} (\bibinfo {year}
  {2005})}\BibitemShut {NoStop}%
\bibitem [{\citenamefont {Bancal}\ \emph {et~al.}(2012)\citenamefont {Bancal},
  \citenamefont {Branciard}, \citenamefont {Brunner}, \citenamefont {Gisin},\
  and\ \citenamefont {Liang}}]{GisinFramework2012}%
  \BibitemOpen
  \bibfield  {author} {\bibinfo {author} {\bibfnamefont {J.-D.}\ \bibnamefont
  {Bancal}}, \bibinfo {author} {\bibfnamefont {C.}~\bibnamefont {Branciard}},
  \bibinfo {author} {\bibfnamefont {N.}~\bibnamefont {Brunner}}, \bibinfo
  {author} {\bibfnamefont {N.}~\bibnamefont {Gisin}}, \ and\ \bibinfo {author}
  {\bibfnamefont {Y.-C.}\ \bibnamefont {Liang}},\ }\bibfield  {title} {\enquote
  {\bibinfo {title} {{A framework for the study of symmetric full-correlation
  {B}ell-like inequalities}},}\ }\href
  {http://stacks.iop.org/1751-8121/45/i=12/a=125301} {\bibfield  {journal}
  {\bibinfo  {journal} {J. Phys. A}\ }\textbf {\bibinfo {volume} {45}},\
  \bibinfo {pages} {125301} (\bibinfo {year} {2012})}\BibitemShut {NoStop}%
\bibitem [{\citenamefont {{Scarani}}(2009)}]{ScaraniNotes2}%
  \BibitemOpen
  \bibfield  {author} {\bibinfo {author} {\bibfnamefont {V.}~\bibnamefont
  {{Scarani}}},\ }\bibfield  {title} {\enquote {\bibinfo {title} {{Quantum
  information: Primitive notions and quantum correlations}},}\ }\href
  {http://arxiv.org/abs/0910.4222} {\bibfield  {journal} {\bibinfo  {journal}
  {arXiv:0910.4222}\ } (\bibinfo {year} {2009})}\BibitemShut {NoStop}%
\bibitem [{\citenamefont {Scarani}(2012)}]{scarani2012device}%
  \BibitemOpen
  \bibfield  {author} {\bibinfo {author} {\bibfnamefont {V.}~\bibnamefont
  {Scarani}},\ }\bibfield  {title} {\enquote {\bibinfo {title} {{The
  device-independent outlook on quantum physics}},}\ }\href@noop {} {\bibfield
  {journal} {\bibinfo  {journal} {Acta Physica Slovaca}\ }\textbf {\bibinfo
  {volume} {62}},\ \bibinfo {pages} {347} (\bibinfo {year} {2012})}\BibitemShut
  {NoStop}%
\bibitem [{\citenamefont {Bancal}(2014)}]{BancalDIApproach}%
  \BibitemOpen
  \bibfield  {author} {\bibinfo {author} {\bibfnamefont {J.-D.}\ \bibnamefont
  {Bancal}},\ }\href {\doibase 10.1007/978-3-319-01183-7} {\emph {\bibinfo
  {title} {On the Device-Independent Approach to Quantum Physics}}}\ (\bibinfo
  {publisher} {Springer International Publishing},\ \bibinfo {year}
  {2014})\BibitemShut {NoStop}%
\bibitem [{\citenamefont {{Chiribella}}\ and\ \citenamefont
  {{Yuan}}(2015)}]{GPTrelatingDI}%
  \BibitemOpen
  \bibfield  {author} {\bibinfo {author} {\bibfnamefont {G.}~\bibnamefont
  {{Chiribella}}}\ and\ \bibinfo {author} {\bibfnamefont {X.}~\bibnamefont
  {{Yuan}}},\ }\bibfield  {title} {\enquote {\bibinfo {title} {{Bridging the
  gap between general probabilistic theories and the device-independent
  framework for nonlocality and contextuality}},}\ }\href
  {http://arxiv.org/abs/1504.02395} {\bibfield  {journal} {\bibinfo  {journal}
  {arXiv:1504.02395}\ } (\bibinfo {year} {2015})}\BibitemShut {NoStop}%
\bibitem [{\citenamefont {Acín}\ \emph {et~al.}(2006)\citenamefont {Acín},
  \citenamefont {Massar},\ and\ \citenamefont {Pironio}}]{CryptoDIQKDNJP}%
  \BibitemOpen
  \bibfield  {author} {\bibinfo {author} {\bibfnamefont {A.}~\bibnamefont
  {Acín}}, \bibinfo {author} {\bibfnamefont {S.}~\bibnamefont {Massar}}, \
  and\ \bibinfo {author} {\bibfnamefont {S.}~\bibnamefont {Pironio}},\
  }\bibfield  {title} {\enquote {\bibinfo {title} {{Efficient quantum key
  distribution secure against no-signalling eavesdroppers}},}\ }\href {\doibase
  10.1088/1367-2630/8/8/126} {\bibfield  {journal} {\bibinfo  {journal} {New J.
  Phys.}\ }\textbf {\bibinfo {volume} {8}},\ \bibinfo {pages} {126} (\bibinfo
  {year} {2006})}\BibitemShut {NoStop}%
\bibitem [{\citenamefont {Masanes}\ \emph {et~al.}(2011)\citenamefont
  {Masanes}, \citenamefont {Pironio},\ and\ \citenamefont
  {Acín}}]{masanesDIQKD}%
  \BibitemOpen
  \bibfield  {author} {\bibinfo {author} {\bibfnamefont {L.}~\bibnamefont
  {Masanes}}, \bibinfo {author} {\bibfnamefont {S.}~\bibnamefont {Pironio}}, \
  and\ \bibinfo {author} {\bibfnamefont {A.}~\bibnamefont {Acín}},\ }\bibfield
   {title} {\enquote {\bibinfo {title} {{Secure device-independent quantum key
  distribution with causally independent measurement devices}},}\ }\href
  {\doibase 10.1038/ncomms1244} {\bibfield  {journal} {\bibinfo  {journal}
  {Nat. Comm.}\ }\textbf {\bibinfo {volume} {2}},\ \bibinfo {pages} {238+}
  (\bibinfo {year} {2011})}\BibitemShut {NoStop}%
\bibitem [{\citenamefont {Ekert}\ and\ \citenamefont
  {Renner}(2014)}]{RenatoQKDNature}%
  \BibitemOpen
  \bibfield  {author} {\bibinfo {author} {\bibfnamefont {A.}~\bibnamefont
  {Ekert}}\ and\ \bibinfo {author} {\bibfnamefont {R.}~\bibnamefont {Renner}},\
  }\bibfield  {title} {\enquote {\bibinfo {title} {{The ultimate physical
  limits of privacy}},}\ }\href {\doibase 10.1038/nature13132} {\bibfield
  {journal} {\bibinfo  {journal} {Nature}\ }\textbf {\bibinfo {volume} {507}},\
  \bibinfo {pages} {443} (\bibinfo {year} {2014})}\BibitemShut {NoStop}%
\bibitem [{\citenamefont {Werner}\ and\ \citenamefont
  {Wolf}(2001)}]{WernerWolfBellInequalities}%
  \BibitemOpen
  \bibfield  {author} {\bibinfo {author} {\bibfnamefont {R.~F.}\ \bibnamefont
  {Werner}}\ and\ \bibinfo {author} {\bibfnamefont {M.~M.}\ \bibnamefont
  {Wolf}},\ }\bibfield  {title} {\enquote {\bibinfo {title} {All-multipartite
  bell-correlation inequalities for two dichotomic observables per site},}\
  }\href {\doibase 10.1103/PhysRevA.64.032112} {\bibfield  {journal} {\bibinfo
  {journal} {Phys. Rev. A}\ }\textbf {\bibinfo {volume} {64}},\ \bibinfo
  {pages} {032112} (\bibinfo {year} {2001})}\BibitemShut {NoStop}%
\bibitem [{\citenamefont {Acín}\ \emph {et~al.}(2015)\citenamefont {Acín},
  \citenamefont {Fritz}, \citenamefont {Leverrier},\ and\ \citenamefont
  {Sainz}}]{FritzCombinatorialLong}%
  \BibitemOpen
  \bibfield  {author} {\bibinfo {author} {\bibfnamefont {A.}~\bibnamefont
  {Acín}}, \bibinfo {author} {\bibfnamefont {T.}~\bibnamefont {Fritz}},
  \bibinfo {author} {\bibfnamefont {A.}~\bibnamefont {Leverrier}}, \ and\
  \bibinfo {author} {\bibfnamefont {A.~B.}\ \bibnamefont {Sainz}},\ }\bibfield
  {title} {\enquote {\bibinfo {title} {{A combinatorial approach to nonlocality
  and contextuality}},}\ }\href {\doibase 10.1007/s00220-014-2260-1} {\bibfield
   {journal} {\bibinfo  {journal} {Comm. Math. Phys.}\ }\textbf {\bibinfo
  {volume} {334}},\ \bibinfo {pages} {533} (\bibinfo {year}
  {2015})}\BibitemShut {NoStop}%
\bibitem [{\citenamefont {Fritz}(2012{\natexlab{a}})}]{KirchbergConjecture}%
  \BibitemOpen
  \bibfield  {author} {\bibinfo {author} {\bibfnamefont {T.}~\bibnamefont
  {Fritz}},\ }\bibfield  {title} {\enquote {\bibinfo {title} {{Tsirelson's
  problem and Kirchberg's conjecture}},}\ }\href {\doibase
  10.1142/S0129055X12500122} {\bibfield  {journal} {\bibinfo  {journal} {Rev.
  Math. Phys.}\ }\textbf {\bibinfo {volume} {24}},\ \bibinfo {pages} {1250012}
  (\bibinfo {year} {2012}{\natexlab{a}})}\BibitemShut {NoStop}%
\bibitem [{\citenamefont {Ozawa}(2013)}]{TsirelsonProblemInteraction}%
  \BibitemOpen
  \bibfield  {author} {\bibinfo {author} {\bibfnamefont {N.}~\bibnamefont
  {Ozawa}},\ }\bibfield  {title} {\enquote {\bibinfo {title} {Tsirelson's
  problem and asymptotically commuting unitary matrices},}\ }\href {\doibase
  10.1063/1.4795391} {\bibfield  {journal} {\bibinfo  {journal} {J. Math.
  Phys.}\ }\textbf {\bibinfo {volume} {54}},\ \bibinfo {eid} {032202} (\bibinfo
  {year} {2013})}\BibitemShut {NoStop}%
\bibitem [{\citenamefont {Avis}\ and\ \citenamefont {Ito}(2007)}]{Avis2007}%
  \BibitemOpen
  \bibfield  {author} {\bibinfo {author} {\bibfnamefont {D.}~\bibnamefont
  {Avis}}\ and\ \bibinfo {author} {\bibfnamefont {T.}~\bibnamefont {Ito}},\
  }\bibfield  {title} {\enquote {\bibinfo {title} {Comparison of two bounds of
  the quantum correlation set},}\ }in\ \href {\doibase 10.1109/icqnm.2007.5}
  {\emph {\bibinfo {booktitle} {1st Inter. Conf. on Quant. Nano \& Micro
  Tech.}}}\ (\bibinfo  {publisher} {{IEEE}},\ \bibinfo {year}
  {2007})\BibitemShut {NoStop}%
\bibitem [{\citenamefont {{Navascués}}\ \emph {et~al.}(2015)\citenamefont
  {{Navascués}}, \citenamefont {{Guryanova}}, \citenamefont {{Hoban}},\ and\
  \citenamefont {{Acín}}}]{AlmostQuantum}%
  \BibitemOpen
  \bibfield  {author} {\bibinfo {author} {\bibfnamefont {M.}~\bibnamefont
  {{Navascués}}}, \bibinfo {author} {\bibfnamefont {Y.}~\bibnamefont
  {{Guryanova}}}, \bibinfo {author} {\bibfnamefont {M.~J.}\ \bibnamefont
  {{Hoban}}}, \ and\ \bibinfo {author} {\bibfnamefont {A.}~\bibnamefont
  {{Acín}}},\ }\bibfield  {title} {\enquote {\bibinfo {title} {{Almost quantum
  correlations}},}\ }\href {\doibase 10.1038/ncomms7288} {\bibfield  {journal}
  {\bibinfo  {journal} {Nat. Commun.}\ }\textbf {\bibinfo {volume} {6}},\
  \bibinfo {pages} {6288} (\bibinfo {year} {2015})}\BibitemShut {NoStop}%
\bibitem [{\citenamefont {Navascués}\ \emph {et~al.}(2007)\citenamefont
  {Navascués}, \citenamefont {Pironio},\ and\ \citenamefont
  {Acín}}]{NPA2007Short}%
  \BibitemOpen
  \bibfield  {author} {\bibinfo {author} {\bibfnamefont {M.}~\bibnamefont
  {Navascués}}, \bibinfo {author} {\bibfnamefont {S.}~\bibnamefont {Pironio}},
  \ and\ \bibinfo {author} {\bibfnamefont {A.}~\bibnamefont {Acín}},\
  }\bibfield  {title} {\enquote {\bibinfo {title} {{Bounding the set of quantum
  correlations}},}\ }\href {\doibase 10.1103/PhysRevLett.98.010401} {\bibfield
  {journal} {\bibinfo  {journal} {Phys. Rev. Lett.}\ }\textbf {\bibinfo
  {volume} {98}},\ \bibinfo {pages} {010401} (\bibinfo {year}
  {2007})}\BibitemShut {NoStop}%
\bibitem [{\citenamefont {Navascués}\ \emph {et~al.}(2008)\citenamefont
  {Navascués}, \citenamefont {Pironio},\ and\ \citenamefont
  {Acín}}]{NPA2008Long}%
  \BibitemOpen
  \bibfield  {author} {\bibinfo {author} {\bibfnamefont {M.}~\bibnamefont
  {Navascués}}, \bibinfo {author} {\bibfnamefont {S.}~\bibnamefont {Pironio}},
  \ and\ \bibinfo {author} {\bibfnamefont {A.}~\bibnamefont {Acín}},\
  }\bibfield  {title} {\enquote {\bibinfo {title} {{A convergent hierarchy of
  semidefinite programs characterizing the set of quantum correlations}},}\
  }\href {\doibase 10.1088/1367-2630/10/7/073013} {\bibfield  {journal}
  {\bibinfo  {journal} {New J. Phys.}\ }\textbf {\bibinfo {volume} {10}},\
  \bibinfo {pages} {073013} (\bibinfo {year} {2008})}\BibitemShut {NoStop}%
\bibitem [{\citenamefont {Pironio}\ \emph {et~al.}(2010)\citenamefont
  {Pironio}, \citenamefont {Navascués},\ and\ \citenamefont
  {Acín}}]{NPAReview}%
  \BibitemOpen
  \bibfield  {author} {\bibinfo {author} {\bibfnamefont {S.}~\bibnamefont
  {Pironio}}, \bibinfo {author} {\bibfnamefont {M.}~\bibnamefont {Navascués}},
  \ and\ \bibinfo {author} {\bibfnamefont {A.}~\bibnamefont {Acín}},\
  }\bibfield  {title} {\enquote {\bibinfo {title} {{Convergent relaxations of
  polynomial optimization problems with noncommuting variables}},}\ }\href
  {\doibase 10.1137/090760155} {\bibfield  {journal} {\bibinfo  {journal} {SIAM
  J. Optim.}\ }\textbf {\bibinfo {volume} {20}},\ \bibinfo {pages} {2157}
  (\bibinfo {year} {2010})}\BibitemShut {NoStop}%
\bibitem [{\citenamefont {Pérez-García}\ \emph {et~al.}(2008)\citenamefont
  {Pérez-García}, \citenamefont {Wolf}, \citenamefont {Palazuelos},
  \citenamefont {Villanueva},\ and\ \citenamefont
  {Junge}}]{DimensionDependantBellViolation}%
  \BibitemOpen
  \bibfield  {author} {\bibinfo {author} {\bibfnamefont {D.}~\bibnamefont
  {Pérez-García}}, \bibinfo {author} {\bibfnamefont {M.}~\bibnamefont
  {Wolf}}, \bibinfo {author} {\bibfnamefont {C.}~\bibnamefont {Palazuelos}},
  \bibinfo {author} {\bibfnamefont {I.}~\bibnamefont {Villanueva}}, \ and\
  \bibinfo {author} {\bibfnamefont {M.}~\bibnamefont {Junge}},\ }\bibfield
  {title} {\enquote {\bibinfo {title} {Unbounded violation of tripartite {B}ell
  inequalities},}\ }\href {\doibase 10.1007/s00220-008-0418-4} {\bibfield
  {journal} {\bibinfo  {journal} {Comm. Math. Phys.}\ }\textbf {\bibinfo
  {volume} {279}},\ \bibinfo {pages} {455} (\bibinfo {year}
  {2008})}\BibitemShut {NoStop}%
\bibitem [{\citenamefont {Junge}\ \emph {et~al.}(2010)\citenamefont {Junge},
  \citenamefont {Palazuelos}, \citenamefont {Pérez-García}, \citenamefont
  {Villanueva},\ and\ \citenamefont {Wolf}}]{DimensionDependantBellViolation2}%
  \BibitemOpen
  \bibfield  {author} {\bibinfo {author} {\bibfnamefont {M.}~\bibnamefont
  {Junge}}, \bibinfo {author} {\bibfnamefont {C.}~\bibnamefont {Palazuelos}},
  \bibinfo {author} {\bibfnamefont {D.}~\bibnamefont {Pérez-García}},
  \bibinfo {author} {\bibfnamefont {I.}~\bibnamefont {Villanueva}}, \ and\
  \bibinfo {author} {\bibfnamefont {M.~M.}\ \bibnamefont {Wolf}},\ }\bibfield
  {title} {\enquote {\bibinfo {title} {Operator space theory: A natural
  framework for {B}ell inequalities},}\ }\href {\doibase
  10.1103/PhysRevLett.104.170405} {\bibfield  {journal} {\bibinfo  {journal}
  {Phys. Rev. Lett.}\ }\textbf {\bibinfo {volume} {104}},\ \bibinfo {pages}
  {170405} (\bibinfo {year} {2010})}\BibitemShut {NoStop}%
\bibitem [{\citenamefont {Epping}\ \emph {et~al.}(2013)\citenamefont {Epping},
  \citenamefont {Kampermann},\ and\ \citenamefont
  {Bru\ss{}}}]{TsirelsonBoundSVD}%
  \BibitemOpen
  \bibfield  {author} {\bibinfo {author} {\bibfnamefont {M.}~\bibnamefont
  {Epping}}, \bibinfo {author} {\bibfnamefont {H.}~\bibnamefont {Kampermann}},
  \ and\ \bibinfo {author} {\bibfnamefont {D.}~\bibnamefont {Bru\ss{}}},\
  }\bibfield  {title} {\enquote {\bibinfo {title} {Designing {B}ell
  inequalities from a {T}sirelson bound},}\ }\href {\doibase
  10.1103/PhysRevLett.111.240404} {\bibfield  {journal} {\bibinfo  {journal}
  {Phys. Rev. Lett.}\ }\textbf {\bibinfo {volume} {111}},\ \bibinfo {pages}
  {240404} (\bibinfo {year} {2013})}\BibitemShut {NoStop}%
\bibitem [{\citenamefont {Żukowski}\ and\ \citenamefont
  {Brukner}(2002)}]{Nqubit.qb}%
  \BibitemOpen
  \bibfield  {author} {\bibinfo {author} {\bibfnamefont {M.}~\bibnamefont
  {Żukowski}}\ and\ \bibinfo {author} {\bibfnamefont {{\v{C}}.}~\bibnamefont
  {Brukner}},\ }\bibfield  {title} {\enquote {\bibinfo {title} {Bell's theorem
  for general \textit{N}-qubit states},}\ }\href {\doibase
  10.1103/PhysRevLett.88.210401} {\bibfield  {journal} {\bibinfo  {journal}
  {Phys. Rev. Lett.}\ }\textbf {\bibinfo {volume} {88}},\ \bibinfo {pages}
  {210401} (\bibinfo {year} {2002})}\BibitemShut {NoStop}%
\bibitem [{\citenamefont {Devetak}\ and\ \citenamefont
  {Winter}(2005)}]{QKD.DD.Winter}%
  \BibitemOpen
  \bibfield  {author} {\bibinfo {author} {\bibfnamefont {I.}~\bibnamefont
  {Devetak}}\ and\ \bibinfo {author} {\bibfnamefont {A.}~\bibnamefont
  {Winter}},\ }\bibfield  {title} {\enquote {\bibinfo {title} {Distillation of
  secret key and entanglement from quantum states},}\ }\href {\doibase
  10.1098/rspa.2004.1372} {\bibfield  {journal} {\bibinfo  {journal} {Proc.
  Roy. Soc. A}\ }\textbf {\bibinfo {volume} {461}},\ \bibinfo {pages} {207}
  (\bibinfo {year} {2005})}\BibitemShut {NoStop}%
\bibitem [{\citenamefont {Kraus}\ \emph {et~al.}(2005)\citenamefont {Kraus},
  \citenamefont {Gisin},\ and\ \citenamefont {Renner}}]{PhysRevLett.95.080501}%
  \BibitemOpen
  \bibfield  {author} {\bibinfo {author} {\bibfnamefont {B.}~\bibnamefont
  {Kraus}}, \bibinfo {author} {\bibfnamefont {N.}~\bibnamefont {Gisin}}, \ and\
  \bibinfo {author} {\bibfnamefont {R.}~\bibnamefont {Renner}},\ }\bibfield
  {title} {\enquote {\bibinfo {title} {Lower and upper bounds on the secret-key
  rate for quantum key distribution protocols using one-way classical
  communication},}\ }\href {\doibase 10.1103/PhysRevLett.95.080501} {\bibfield
  {journal} {\bibinfo  {journal} {Phys. Rev. Lett.}\ }\textbf {\bibinfo
  {volume} {95}},\ \bibinfo {pages} {080501} (\bibinfo {year}
  {2005})}\BibitemShut {NoStop}%
\bibitem [{\citenamefont {Brunner}\ \emph {et~al.}(2008)\citenamefont
  {Brunner}, \citenamefont {Pironio}, \citenamefont {Acin}, \citenamefont
  {Gisin}, \citenamefont {Méthot},\ and\ \citenamefont
  {Scarani}}]{DimensionWitness2008Brunner}%
  \BibitemOpen
  \bibfield  {author} {\bibinfo {author} {\bibfnamefont {N.}~\bibnamefont
  {Brunner}}, \bibinfo {author} {\bibfnamefont {S.}~\bibnamefont {Pironio}},
  \bibinfo {author} {\bibfnamefont {A.}~\bibnamefont {Acin}}, \bibinfo {author}
  {\bibfnamefont {N.}~\bibnamefont {Gisin}}, \bibinfo {author} {\bibfnamefont
  {A.~A.}\ \bibnamefont {Méthot}}, \ and\ \bibinfo {author} {\bibfnamefont
  {V.}~\bibnamefont {Scarani}},\ }\bibfield  {title} {\enquote {\bibinfo
  {title} {Testing the dimension of {H}ilbert spaces},}\ }\href {\doibase
  10.1103/PhysRevLett.100.210503} {\bibfield  {journal} {\bibinfo  {journal}
  {Phys. Rev. Lett.}\ }\textbf {\bibinfo {volume} {100}},\ \bibinfo {pages}
  {210503} (\bibinfo {year} {2008})}\BibitemShut {NoStop}%
\bibitem [{\citenamefont {Navascués}\ \emph {et~al.}(2014)\citenamefont
  {Navascués}, \citenamefont {de~la Torre},\ and\ \citenamefont
  {Vértesi}}]{Navascues2013FiniteDimensions}%
  \BibitemOpen
  \bibfield  {author} {\bibinfo {author} {\bibfnamefont {M.}~\bibnamefont
  {Navascués}}, \bibinfo {author} {\bibfnamefont {G.}~\bibnamefont {de~la
  Torre}}, \ and\ \bibinfo {author} {\bibfnamefont {T.}~\bibnamefont
  {Vértesi}},\ }\bibfield  {title} {\enquote {\bibinfo {title}
  {Characterization of quantum correlations with local dimension constraints
  and its device-independent applications},}\ }\href {\doibase
  10.1103/PhysRevX.4.011011} {\bibfield  {journal} {\bibinfo  {journal} {Phys.
  Rev. X}\ }\textbf {\bibinfo {volume} {4}},\ \bibinfo {pages} {011011}
  (\bibinfo {year} {2014})}\BibitemShut {NoStop}%
\bibitem [{\citenamefont {Navascués}\ and\ \citenamefont
  {Vértesi}(2015)}]{Navascues2014FiniteDimensions}%
  \BibitemOpen
  \bibfield  {author} {\bibinfo {author} {\bibfnamefont {M.}~\bibnamefont
  {Navascués}}\ and\ \bibinfo {author} {\bibfnamefont {T.}~\bibnamefont
  {Vértesi}},\ }\bibfield  {title} {\enquote {\bibinfo {title} {{Bounding the
  set of finite dimensional quantum correlations}},}\ }\href {\doibase
  10.1103/PhysRevLett.115.020501} {\bibfield  {journal} {\bibinfo  {journal}
  {Phys. Rev. Lett.}\ }\textbf {\bibinfo {volume} {115}},\ \bibinfo {pages}
  {020501} (\bibinfo {year} {2015})}\BibitemShut {NoStop}%
\bibitem [{\citenamefont {Bowles}\ \emph {et~al.}(2014)\citenamefont {Bowles},
  \citenamefont {Quintino},\ and\ \citenamefont
  {Brunner}}]{ClassicalDimensionBoundNonConvex}%
  \BibitemOpen
  \bibfield  {author} {\bibinfo {author} {\bibfnamefont {J.}~\bibnamefont
  {Bowles}}, \bibinfo {author} {\bibfnamefont {M.~T.}\ \bibnamefont
  {Quintino}}, \ and\ \bibinfo {author} {\bibfnamefont {N.}~\bibnamefont
  {Brunner}},\ }\bibfield  {title} {\enquote {\bibinfo {title} {Certifying the
  dimension of classical and quantum systems in a prepare-and-measure scenario
  with independent devices},}\ }\href {\doibase 10.1103/PhysRevLett.112.140407}
  {\bibfield  {journal} {\bibinfo  {journal} {Phys. Rev. Lett.}\ }\textbf
  {\bibinfo {volume} {112}},\ \bibinfo {pages} {140407} (\bibinfo {year}
  {2014})}\BibitemShut {NoStop}%
\bibitem [{\citenamefont {Gallego}\ \emph {et~al.}(2010)\citenamefont
  {Gallego}, \citenamefont {Brunner}, \citenamefont {Hadley},\ and\
  \citenamefont {Acín}}]{ClassicalDimensionBound2010}%
  \BibitemOpen
  \bibfield  {author} {\bibinfo {author} {\bibfnamefont {R.}~\bibnamefont
  {Gallego}}, \bibinfo {author} {\bibfnamefont {N.}~\bibnamefont {Brunner}},
  \bibinfo {author} {\bibfnamefont {C.}~\bibnamefont {Hadley}}, \ and\ \bibinfo
  {author} {\bibfnamefont {A.}~\bibnamefont {Acín}},\ }\bibfield  {title}
  {\enquote {\bibinfo {title} {Device-independent tests of classical and
  quantum dimensions},}\ }\href {\doibase 10.1103/PhysRevLett.105.230501}
  {\bibfield  {journal} {\bibinfo  {journal} {Phys. Rev. Lett.}\ }\textbf
  {\bibinfo {volume} {105}},\ \bibinfo {pages} {230501} (\bibinfo {year}
  {2010})}\BibitemShut {NoStop}%
\bibitem [{\citenamefont {Jain}\ \emph {et~al.}(2013)\citenamefont {Jain},
  \citenamefont {Shi}, \citenamefont {Wei},\ and\ \citenamefont
  {Zhang}}]{PSDasQuantum}%
  \BibitemOpen
  \bibfield  {author} {\bibinfo {author} {\bibfnamefont {R.}~\bibnamefont
  {Jain}}, \bibinfo {author} {\bibfnamefont {Y.}~\bibnamefont {Shi}}, \bibinfo
  {author} {\bibfnamefont {Z.}~\bibnamefont {Wei}}, \ and\ \bibinfo {author}
  {\bibfnamefont {S.}~\bibnamefont {Zhang}},\ }\bibfield  {title} {\enquote
  {\bibinfo {title} {Efficient protocols for generating bipartite classical
  distributions and quantum states},}\ }\href {\doibase
  10.1109/TIT.2013.2258372} {\bibfield  {journal} {\bibinfo  {journal} {IEEE
  Trans. Info. Theo.}\ }\textbf {\bibinfo {volume} {59}},\ \bibinfo {pages}
  {5171} (\bibinfo {year} {2013})}\BibitemShut {NoStop}%
\bibitem [{\citenamefont {{Fawzi}}\ \emph {et~al.}(2015)\citenamefont
  {{Fawzi}}, \citenamefont {{Gouveia}}, \citenamefont {{Parrilo}},
  \citenamefont {{Robinson}},\ and\ \citenamefont {{Thomas}}}]{ReviewPSD}%
  \BibitemOpen
  \bibfield  {author} {\bibinfo {author} {\bibfnamefont {H.}~\bibnamefont
  {{Fawzi}}}, \bibinfo {author} {\bibfnamefont {J.}~\bibnamefont {{Gouveia}}},
  \bibinfo {author} {\bibfnamefont {P.~A.}\ \bibnamefont {{Parrilo}}}, \bibinfo
  {author} {\bibfnamefont {R.~Z.}\ \bibnamefont {{Robinson}}}, \ and\ \bibinfo
  {author} {\bibfnamefont {R.~R.}\ \bibnamefont {{Thomas}}},\ }\bibfield
  {title} {\enquote {\bibinfo {title} {Positive semidefinite rank},}\ }\href
  {http://link.springer.com/article/10.1007\%2Fs10107-015-0922-1} {\bibfield
  {journal} {\bibinfo  {journal} {Math. Program.}\ }\textbf {\bibinfo {volume}
  {153}},\ \bibinfo {pages} {133} (\bibinfo {year} {2015})}\BibitemShut
  {NoStop}%
\bibitem [{\citenamefont {Zhang}(2012)}]{ZhangQuantumGameTheory}%
  \BibitemOpen
  \bibfield  {author} {\bibinfo {author} {\bibfnamefont {S.}~\bibnamefont
  {Zhang}},\ }\bibfield  {title} {\enquote {\bibinfo {title} {Quantum strategic
  game theory},}\ }in\ \href {http://arxiv.org/abs/1012.5141} {\emph {\bibinfo
  {booktitle} {Proc. 3rd Innov. Theo. Comput. Sci. - {ITCS}
  {\textquoteright}12}}}\ (\bibinfo  {publisher} {ACM},\ \bibinfo {address}
  {New York, NY, USA},\ \bibinfo {year} {2012})\ pp.\ \bibinfo {pages}
  {39--59},\ \bibinfo {note} {429122}\BibitemShut {NoStop}%
\bibitem [{\citenamefont {{Sikora}}\ \emph {et~al.}(2015)\citenamefont
  {{Sikora}}, \citenamefont {{Varvitsiotis}},\ and\ \citenamefont
  {{Wei}}}]{ConditionalMinimumHSD}%
  \BibitemOpen
  \bibfield  {author} {\bibinfo {author} {\bibfnamefont {J.}~\bibnamefont
  {{Sikora}}}, \bibinfo {author} {\bibfnamefont {A.}~\bibnamefont
  {{Varvitsiotis}}}, \ and\ \bibinfo {author} {\bibfnamefont {Z.}~\bibnamefont
  {{Wei}}},\ }\bibfield  {title} {\enquote {\bibinfo {title} {{On the minimum
  dimension of a Hilbert space needed to generate a quantum correlation}},}\
  }\href {http://arxiv.org/abs/1507.00213} {\bibfield  {journal} {\bibinfo
  {journal} {arXiv:1507.00213}\ } (\bibinfo {year} {2015})}\BibitemShut
  {NoStop}%
\bibitem [{\citenamefont {{Harrigan}}\ \emph {et~al.}(2007)\citenamefont
  {{Harrigan}}, \citenamefont {{Rudolph}},\ and\ \citenamefont
  {{Aaronson}}}]{ClassicalLambdaRequirement}%
  \BibitemOpen
  \bibfield  {author} {\bibinfo {author} {\bibfnamefont {N.}~\bibnamefont
  {{Harrigan}}}, \bibinfo {author} {\bibfnamefont {T.}~\bibnamefont
  {{Rudolph}}}, \ and\ \bibinfo {author} {\bibfnamefont {S.}~\bibnamefont
  {{Aaronson}}},\ }\bibfield  {title} {\enquote {\bibinfo {title}
  {{Representing probabilistic data via ontological models}},}\ }\href
  {http://arxiv.org/abs/0709.1149} {\bibfield  {journal} {\bibinfo  {journal}
  {arXiv:0709.1149}\ } (\bibinfo {year} {2007})}\BibitemShut {NoStop}%
\bibitem [{\citenamefont {Wehner}\ \emph {et~al.}(2008)\citenamefont {Wehner},
  \citenamefont {Christandl},\ and\ \citenamefont
  {Doherty}}]{DimensionWitness2008Wehner}%
  \BibitemOpen
  \bibfield  {author} {\bibinfo {author} {\bibfnamefont {S.}~\bibnamefont
  {Wehner}}, \bibinfo {author} {\bibfnamefont {M.}~\bibnamefont {Christandl}},
  \ and\ \bibinfo {author} {\bibfnamefont {A.~C.}\ \bibnamefont {Doherty}},\
  }\bibfield  {title} {\enquote {\bibinfo {title} {Lower bound on the dimension
  of a quantum system given measured data},}\ }\href {\doibase
  10.1103/PhysRevA.78.062112} {\bibfield  {journal} {\bibinfo  {journal} {Phys.
  Rev. A}\ }\textbf {\bibinfo {volume} {78}},\ \bibinfo {pages} {062112}
  (\bibinfo {year} {2008})}\BibitemShut {NoStop}%
\bibitem [{\citenamefont {Pál}\ and\ \citenamefont
  {Vértesi}(2009)}]{ConcavityPaperVertesi}%
  \BibitemOpen
  \bibfield  {author} {\bibinfo {author} {\bibfnamefont {K.~F.}\ \bibnamefont
  {Pál}}\ and\ \bibinfo {author} {\bibfnamefont {T.}~\bibnamefont
  {Vértesi}},\ }\bibfield  {title} {\enquote {\bibinfo {title} {Concavity of
  the set of quantum probabilities for any given dimension},}\ }\href {\doibase
  10.1103/PhysRevA.80.042114} {\bibfield  {journal} {\bibinfo  {journal} {Phys.
  Rev. A}\ }\textbf {\bibinfo {volume} {80}},\ \bibinfo {pages} {042114}
  (\bibinfo {year} {2009})}\BibitemShut {NoStop}%
\bibitem [{\citenamefont {Clauser}\ \emph {et~al.}(1969)\citenamefont
  {Clauser}, \citenamefont {Horne}, \citenamefont {Shimony},\ and\
  \citenamefont {Holt}}]{CHSHOriginal}%
  \BibitemOpen
  \bibfield  {author} {\bibinfo {author} {\bibfnamefont {J.~F.}\ \bibnamefont
  {Clauser}}, \bibinfo {author} {\bibfnamefont {M.~A.}\ \bibnamefont {Horne}},
  \bibinfo {author} {\bibfnamefont {A.}~\bibnamefont {Shimony}}, \ and\
  \bibinfo {author} {\bibfnamefont {R.~A.}\ \bibnamefont {Holt}},\ }\bibfield
  {title} {\enquote {\bibinfo {title} {{Proposed experiment to test local
  hidden-variable theories}},}\ }\href {\doibase 10.1103/PhysRevLett.23.880}
  {\bibfield  {journal} {\bibinfo  {journal} {Phys. Rev. Lett.}\ }\textbf
  {\bibinfo {volume} {23}},\ \bibinfo {pages} {880} (\bibinfo {year}
  {1969})}\BibitemShut {NoStop}%
\bibitem [{\citenamefont {Wolfe}\ and\ \citenamefont {Yelin}(2012)}]{WolfeQB}%
  \BibitemOpen
  \bibfield  {author} {\bibinfo {author} {\bibfnamefont {E.}~\bibnamefont
  {Wolfe}}\ and\ \bibinfo {author} {\bibfnamefont {S.~F.}\ \bibnamefont
  {Yelin}},\ }\bibfield  {title} {\enquote {\bibinfo {title} {{Quantum bounds
  for inequalities involving marginal expectation values}},}\ }\href {\doibase
  10.1103/Physreva.86.012123} {\bibfield  {journal} {\bibinfo  {journal} {Phys.
  Rev. A}\ }\textbf {\bibinfo {volume} {86}},\ \bibinfo {pages} {012123}
  (\bibinfo {year} {2012})}\BibitemShut {NoStop}%
\bibitem [{\citenamefont {Piani}\ \emph {et~al.}(2008)\citenamefont {Piani},
  \citenamefont {Horodecki},\ and\ \citenamefont {Horodecki}}]{CQstate2008}%
  \BibitemOpen
  \bibfield  {author} {\bibinfo {author} {\bibfnamefont {M.}~\bibnamefont
  {Piani}}, \bibinfo {author} {\bibfnamefont {P.}~\bibnamefont {Horodecki}}, \
  and\ \bibinfo {author} {\bibfnamefont {R.}~\bibnamefont {Horodecki}},\
  }\bibfield  {title} {\enquote {\bibinfo {title} {No-local-broadcasting
  theorem for multipartite quantum correlations},}\ }\href {\doibase
  10.1103/PhysRevLett.100.090502} {\bibfield  {journal} {\bibinfo  {journal}
  {Phys. Rev. Lett.}\ }\textbf {\bibinfo {volume} {100}},\ \bibinfo {pages}
  {090502} (\bibinfo {year} {2008})}\BibitemShut {NoStop}%
\bibitem [{\citenamefont {Acín}\ \emph {et~al.}(2005)\citenamefont {Acín},
  \citenamefont {Gill},\ and\ \citenamefont {Gisin}}]{ExplicitNoSigBasis}%
  \BibitemOpen
  \bibfield  {author} {\bibinfo {author} {\bibfnamefont {A.}~\bibnamefont
  {Acín}}, \bibinfo {author} {\bibfnamefont {R.}~\bibnamefont {Gill}}, \ and\
  \bibinfo {author} {\bibfnamefont {N.}~\bibnamefont {Gisin}},\ }\bibfield
  {title} {\enquote {\bibinfo {title} {Optimal {B}ell tests do not require
  maximally entangled states},}\ }\href {\doibase
  10.1103/PhysRevLett.95.210402} {\bibfield  {journal} {\bibinfo  {journal}
  {Phys. Rev. Lett.}\ }\textbf {\bibinfo {volume} {95}},\ \bibinfo {pages}
  {210402} (\bibinfo {year} {2005})}\BibitemShut {NoStop}%
\bibitem [{\citenamefont {Pironio}(2005)}]{PironioDimension}%
  \BibitemOpen
  \bibfield  {author} {\bibinfo {author} {\bibfnamefont {S.}~\bibnamefont
  {Pironio}},\ }\bibfield  {title} {\enquote {\bibinfo {title} {{Lifting {B}ell
  inequalities}},}\ }\href {http://dx.doi.org/10.1063/1.1928727} {\bibfield
  {journal} {\bibinfo  {journal} {J. Math. Phys.}\ }\textbf {\bibinfo {volume}
  {46}},\ \bibinfo {eid} {062112} (\bibinfo {year} {2005})}\BibitemShut
  {NoStop}%
\bibitem [{\citenamefont {Fine}(1982)}]{FineTheorem}%
  \BibitemOpen
  \bibfield  {author} {\bibinfo {author} {\bibfnamefont {A.}~\bibnamefont
  {Fine}},\ }\bibfield  {title} {\enquote {\bibinfo {title} {Hidden variables,
  joint probability, and the {B}ell inequalities},}\ }\href {\doibase
  10.1103/PhysRevLett.48.291} {\bibfield  {journal} {\bibinfo  {journal} {Phys.
  Rev. Lett.}\ }\textbf {\bibinfo {volume} {48}},\ \bibinfo {pages} {291}
  (\bibinfo {year} {1982})}\BibitemShut {NoStop}%
\bibitem [{\citenamefont {Spekkens}(2014)}]{SpekkensDeterminism}%
  \BibitemOpen
  \bibfield  {author} {\bibinfo {author} {\bibfnamefont {R.}~\bibnamefont
  {Spekkens}},\ }\bibfield  {title} {\enquote {\bibinfo {title} {The status of
  determinism in proofs of the impossibility of a noncontextual model of
  quantum theory},}\ }\href {\doibase 10.1007/s10701-014-9833-x} {\bibfield
  {journal} {\bibinfo  {journal} {Found. Phys.}\ }\textbf {\bibinfo {volume}
  {44}},\ \bibinfo {pages} {1125} (\bibinfo {year} {2014})}\BibitemShut
  {NoStop}%
\bibitem [{\citenamefont {Kunjwal}(2015)}]{KunjwalFineTheorem}%
  \BibitemOpen
  \bibfield  {author} {\bibinfo {author} {\bibfnamefont {R.}~\bibnamefont
  {Kunjwal}},\ }\bibfield  {title} {\enquote {\bibinfo {title} {Fine's theorem,
  noncontextuality, and correlations in {S}pecker's scenario},}\ }\href
  {\doibase 10.1103/PhysRevA.91.022108} {\bibfield  {journal} {\bibinfo
  {journal} {Phys. Rev. A}\ }\textbf {\bibinfo {volume} {91}},\ \bibinfo
  {pages} {022108} (\bibinfo {year} {2015})}\BibitemShut {NoStop}%
\bibitem [{\citenamefont {Bárány}(1982)}]{Caratheodory}%
  \BibitemOpen
  \bibfield  {author} {\bibinfo {author} {\bibfnamefont {I.}~\bibnamefont
  {Bárány}},\ }\bibfield  {title} {\enquote {\bibinfo {title} {A
  generalization of {C}arathéodory's theorem},}\ }\href {\doibase
  10.1016/0012-365X(82)90115-7} {\bibfield  {journal} {\bibinfo  {journal}
  {Discrete Mathematics}\ }\textbf {\bibinfo {volume} {40}},\ \bibinfo {pages}
  {141 } (\bibinfo {year} {1982})}\BibitemShut {NoStop}%
\bibitem [{\citenamefont
  {Weisstein}(2015{\natexlab{a}})}]{CaratheodoryMathworld}%
  \BibitemOpen
  \bibfield  {author} {\bibinfo {author} {\bibfnamefont {E.~W.}\ \bibnamefont
  {Weisstein}},\ }\href
  {http://mathworld.wolfram.com/CaratheodorysFundamentalTheorem.html} {\enquote
  {\bibinfo {title} {Carathéodory's fundamental theorem},}\ } (\bibinfo {year}
  {2015}{\natexlab{a}}),\ \bibinfo {note} {see also
  \href[pdfnewwindow]{http://mathoverflow.net/questions/77379\#77382}{mathoverflow.net/q/77379}}\BibitemShut
  {NoStop}%
\bibitem [{\citenamefont
  {Weisstein}(2015{\natexlab{b}})}]{StirlingS2Mathworld}%
  \BibitemOpen
  \bibfield  {author} {\bibinfo {author} {\bibfnamefont {E.~W.}\ \bibnamefont
  {Weisstein}},\ }\href
  {http://mathworld.wolfram.com/StirlingNumberoftheSecondKind.html} {\enquote
  {\bibinfo {title} {Stirling number of the second kind},}\ } (\bibinfo {year}
  {2015}{\natexlab{b}}),\ \bibinfo {note} {see also
  \href[pdfnewwindow]{http://math.stackexchange.com/a/264814}{math.stackexchange.com/a/264814}}\BibitemShut
  {NoStop}%
\bibitem [{\citenamefont {Horodecki}\ \emph {et~al.}(2001)\citenamefont
  {Horodecki}, \citenamefont {Horodecki},\ and\ \citenamefont
  {Horodecki}}]{NSepStateMixing}%
  \BibitemOpen
  \bibfield  {author} {\bibinfo {author} {\bibfnamefont {M.}~\bibnamefont
  {Horodecki}}, \bibinfo {author} {\bibfnamefont {P.}~\bibnamefont
  {Horodecki}}, \ and\ \bibinfo {author} {\bibfnamefont {R.}~\bibnamefont
  {Horodecki}},\ }\bibfield  {title} {\enquote {\bibinfo {title} {Separability
  of n-particle mixed states: necessary and sufficient conditions in terms of
  linear maps},}\ }\href {\doibase 10.1016/S0375-9601(01)00142-6} {\bibfield
  {journal} {\bibinfo  {journal} {Phys. Lett. A}\ }\textbf {\bibinfo {volume}
  {283}},\ \bibinfo {pages} {1 } (\bibinfo {year} {2001})}\BibitemShut
  {NoStop}%
\bibitem [{\citenamefont {Lockhart}(2000)}]{SeparableEnsembleLength}%
  \BibitemOpen
  \bibfield  {author} {\bibinfo {author} {\bibfnamefont {R.}~\bibnamefont
  {Lockhart}},\ }\bibfield  {title} {\enquote {\bibinfo {title} {Optimal
  ensemble length of mixed separable states},}\ }\href {\doibase
  10.1063/1.1290055} {\bibfield  {journal} {\bibinfo  {journal} {J. of Math.
  Phys.}\ }\textbf {\bibinfo {volume} {41}},\ \bibinfo {pages} {6766} (\bibinfo
  {year} {2000})}\BibitemShut {NoStop}%
\bibitem [{\citenamefont {Vértesi}\ and\ \citenamefont
  {Bene}(2010)}]{VBInequality}%
  \BibitemOpen
  \bibfield  {author} {\bibinfo {author} {\bibfnamefont {T.}~\bibnamefont
  {Vértesi}}\ and\ \bibinfo {author} {\bibfnamefont {E.}~\bibnamefont
  {Bene}},\ }\bibfield  {title} {\enquote {\bibinfo {title} {Two-qubit {B}ell
  inequality for which positive operator-valued measurements are relevant},}\
  }\href {\doibase 10.1103/PhysRevA.82.062115} {\bibfield  {journal} {\bibinfo
  {journal} {Phys. Rev. A}\ }\textbf {\bibinfo {volume} {82}},\ \bibinfo
  {pages} {062115} (\bibinfo {year} {2010})}\BibitemShut {NoStop}%
\bibitem [{\citenamefont {Groisman}\ \emph {et~al.}(2005)\citenamefont
  {Groisman}, \citenamefont {Popescu},\ and\ \citenamefont
  {Winter}}]{PopescuTotalCorrelation}%
  \BibitemOpen
  \bibfield  {author} {\bibinfo {author} {\bibfnamefont {B.}~\bibnamefont
  {Groisman}}, \bibinfo {author} {\bibfnamefont {S.}~\bibnamefont {Popescu}}, \
  and\ \bibinfo {author} {\bibfnamefont {A.}~\bibnamefont {Winter}},\
  }\bibfield  {title} {\enquote {\bibinfo {title} {Quantum, classical, and
  total amount of correlations in a quantum state},}\ }\href {\doibase
  10.1103/PhysRevA.72.032317} {\bibfield  {journal} {\bibinfo  {journal} {Phys.
  Rev. A}\ }\textbf {\bibinfo {volume} {72}},\ \bibinfo {pages} {032317}
  (\bibinfo {year} {2005})}\BibitemShut {NoStop}%
\bibitem [{\citenamefont {Cirel'son}(1980)}]{Tsirelson1980}%
  \BibitemOpen
  \bibfield  {author} {\bibinfo {author} {\bibfnamefont {B.~S.}\ \bibnamefont
  {Cirel'son}},\ }\bibfield  {title} {\enquote {\bibinfo {title} {{Quantum
  generalizations of {B}ell's inequality}},}\ }\href {\doibase
  10.1007/BF00417500} {\bibfield  {journal} {\bibinfo  {journal} {Lett. Math.
  Phys.}\ }\textbf {\bibinfo {volume} {4}},\ \bibinfo {pages} {93} (\bibinfo
  {year} {1980})}\BibitemShut {NoStop}%
\bibitem [{\citenamefont {Acín}\ \emph {et~al.}(2000)\citenamefont {Acín},
  \citenamefont {Andrianov}, \citenamefont {Costa}, \citenamefont {Jané},
  \citenamefont {Latorre},\ and\ \citenamefont
  {Tarrach}}]{acin2000generalized}%
  \BibitemOpen
  \bibfield  {author} {\bibinfo {author} {\bibfnamefont {A.}~\bibnamefont
  {Acín}}, \bibinfo {author} {\bibfnamefont {A.}~\bibnamefont {Andrianov}},
  \bibinfo {author} {\bibfnamefont {L.}~\bibnamefont {Costa}}, \bibinfo
  {author} {\bibfnamefont {E.}~\bibnamefont {Jané}}, \bibinfo {author}
  {\bibfnamefont {J.~I.}\ \bibnamefont {Latorre}}, \ and\ \bibinfo {author}
  {\bibfnamefont {R.}~\bibnamefont {Tarrach}},\ }\bibfield  {title} {\enquote
  {\bibinfo {title} {Generalized {S}chmidt decomposition and classification of
  three-quantum-bit states},}\ }\href {\doibase 10.1103/PhysRevLett.85.1560}
  {\bibfield  {journal} {\bibinfo  {journal} {Phys. Rev. Lett.}\ }\textbf
  {\bibinfo {volume} {85}},\ \bibinfo {pages} {1560} (\bibinfo {year}
  {2000})}\BibitemShut {NoStop}%
\bibitem [{\citenamefont {Bruß}(2002)}]{characterizingentanglement}%
  \BibitemOpen
  \bibfield  {author} {\bibinfo {author} {\bibfnamefont {D.}~\bibnamefont
  {Bruß}},\ }\bibfield  {title} {\enquote {\bibinfo {title} {{Characterizing
  entanglement}},}\ }\href {\doibase 10.1063/1.1494474} {\bibfield  {journal}
  {\bibinfo  {journal} {J. Math. Phys.}\ }\textbf {\bibinfo {volume} {43}},\
  \bibinfo {pages} {4237} (\bibinfo {year} {2002})}\BibitemShut {NoStop}%
\bibitem [{\citenamefont {Amico}\ \emph {et~al.}(2008)\citenamefont {Amico},
  \citenamefont {Osterloh},\ and\ \citenamefont {Vedral}}]{multireview}%
  \BibitemOpen
  \bibfield  {author} {\bibinfo {author} {\bibfnamefont {L.}~\bibnamefont
  {Amico}}, \bibinfo {author} {\bibfnamefont {A.}~\bibnamefont {Osterloh}}, \
  and\ \bibinfo {author} {\bibfnamefont {V.}~\bibnamefont {Vedral}},\
  }\bibfield  {title} {\enquote {\bibinfo {title} {{Entanglement in Many-Body
  Systems}},}\ }\href {\doibase 10.1103/revmodphys.80.517} {\bibfield
  {journal} {\bibinfo  {journal} {Rev. Mod. Phys.}\ }\textbf {\bibinfo {volume}
  {80}},\ \bibinfo {pages} {517} (\bibinfo {year} {2008})}\BibitemShut
  {NoStop}%
\bibitem [{\citenamefont {Gühne}\ and\ \citenamefont
  {Tóth}(2009)}]{entang.review.toth}%
  \BibitemOpen
  \bibfield  {author} {\bibinfo {author} {\bibfnamefont {O.}~\bibnamefont
  {Gühne}}\ and\ \bibinfo {author} {\bibfnamefont {G.}~\bibnamefont {Tóth}},\
  }\bibfield  {title} {\enquote {\bibinfo {title} {{Entanglement detection}},}\
  }\href {\doibase 10.1016/j.physrep.2009.02.004} {\bibfield  {journal}
  {\bibinfo  {journal} {Phys. Rep.}\ }\textbf {\bibinfo {volume} {474}},\
  \bibinfo {pages} {1} (\bibinfo {year} {2009})}\BibitemShut {NoStop}%
\bibitem [{\citenamefont {Heinosaari}\ \emph {et~al.}(2008)\citenamefont
  {Heinosaari}, \citenamefont {Reitzner},\ and\ \citenamefont
  {Stano}}]{HeinosaariPOVM}%
  \BibitemOpen
  \bibfield  {author} {\bibinfo {author} {\bibfnamefont {T.}~\bibnamefont
  {Heinosaari}}, \bibinfo {author} {\bibfnamefont {D.}~\bibnamefont
  {Reitzner}}, \ and\ \bibinfo {author} {\bibfnamefont {P.}~\bibnamefont
  {Stano}},\ }\bibfield  {title} {\enquote {\bibinfo {title} {Notes on joint
  measurability of quantum observables},}\ }\href {\doibase
  10.1007/s10701-008-9256-7} {\bibfield  {journal} {\bibinfo  {journal} {Found.
  Phys.}\ }\textbf {\bibinfo {volume} {38}},\ \bibinfo {pages} {1133} (\bibinfo
  {year} {2008})}\BibitemShut {NoStop}%
\bibitem [{\citenamefont {Junge}\ \emph {et~al.}(2011)\citenamefont {Junge},
  \citenamefont {Navascués}, \citenamefont {Palazuelos}, \citenamefont
  {Perez-Garcia}, \citenamefont {Scholz},\ and\ \citenamefont
  {Werner}}]{ConnesEmbedding}%
  \BibitemOpen
  \bibfield  {author} {\bibinfo {author} {\bibfnamefont {M.}~\bibnamefont
  {Junge}}, \bibinfo {author} {\bibfnamefont {M.}~\bibnamefont {Navascués}},
  \bibinfo {author} {\bibfnamefont {C.}~\bibnamefont {Palazuelos}}, \bibinfo
  {author} {\bibfnamefont {D.}~\bibnamefont {Perez-Garcia}}, \bibinfo {author}
  {\bibfnamefont {V.~B.}\ \bibnamefont {Scholz}}, \ and\ \bibinfo {author}
  {\bibfnamefont {R.~F.}\ \bibnamefont {Werner}},\ }\bibfield  {title}
  {\enquote {\bibinfo {title} {{Connes's embedding problem and {T}sirelson's
  problem}},}\ }\href
  {http://scitation.aip.org/content/aip/journal/jmp/52/1/10.1063/1.3514538}
  {\bibfield  {journal} {\bibinfo  {journal} {J. Math. Phys.}\ }\textbf
  {\bibinfo {volume} {52}},\ \bibinfo {eid} {012102} (\bibinfo {year}
  {2011})}\BibitemShut {NoStop}%
\bibitem [{\citenamefont {Navascués}\ \emph {et~al.}(2012)\citenamefont
  {Navascués}, \citenamefont {Cooney}, \citenamefont {Pérez-García},\ and\
  \citenamefont {Villanueva}}]{PhysicalTsirelson}%
  \BibitemOpen
  \bibfield  {author} {\bibinfo {author} {\bibfnamefont {M.}~\bibnamefont
  {Navascués}}, \bibinfo {author} {\bibfnamefont {T.}~\bibnamefont {Cooney}},
  \bibinfo {author} {\bibfnamefont {D.}~\bibnamefont {Pérez-García}}, \ and\
  \bibinfo {author} {\bibfnamefont {N.}~\bibnamefont {Villanueva}},\ }\bibfield
   {title} {\enquote {\bibinfo {title} {{A physical approach to {T}sirelson’s
  problem}},}\ }\href {\doibase 10.1007/s10701-012-9641-0} {\bibfield
  {journal} {\bibinfo  {journal} {Found. Phys.}\ }\textbf {\bibinfo {volume}
  {42}},\ \bibinfo {pages} {985} (\bibinfo {year} {2012})}\BibitemShut
  {NoStop}%
\bibitem [{\citenamefont {{Scholz}}\ and\ \citenamefont
  {{Werner}}(2008)}]{TsirelsonProblemArXiv}%
  \BibitemOpen
  \bibfield  {author} {\bibinfo {author} {\bibfnamefont {V.~B.}\ \bibnamefont
  {{Scholz}}}\ and\ \bibinfo {author} {\bibfnamefont {R.~F.}\ \bibnamefont
  {{Werner}}},\ }\bibfield  {title} {\enquote {\bibinfo {title} {{Tsirelson's
  problem}},}\ }\href {http://arxiv.org/abs/0812.4305} {\bibfield  {journal}
  {\bibinfo  {journal} {arXiv:0812.4305}\ } (\bibinfo {year}
  {2008})}\BibitemShut {NoStop}%
\bibitem [{loc()}]{localsom}%
  \BibitemOpen
  \href@noop {} {}\bibinfo {note} {See Supplemental Material in arXiv source at
  \href[pdfnewwindow]{http://arxiv.org/format/1506.01119}{arXiv:1506.01119}}\BibitemShut
  {NoStop}%
\bibitem [{\citenamefont {Goh}\ \emph {et~al.}(2015)\citenamefont {Goh},
  \citenamefont {Bancal},\ and\ \citenamefont {Scarani}}]{ScaraniUnpublished}%
  \BibitemOpen
  \bibfield  {author} {\bibinfo {author} {\bibfnamefont {K.~T.}\ \bibnamefont
  {Goh}}, \bibinfo {author} {\bibfnamefont {J.-D.}\ \bibnamefont {Bancal}}, \
  and\ \bibinfo {author} {\bibfnamefont {V.}~\bibnamefont {Scarani}},\
  }\bibfield  {title} {\enquote {\bibinfo {title} {"certification of
  entanglement based only on hilbert space dimension"},}\ }\href
  {http://arxiv.org/abs/1509.08682} {\bibfield  {journal} {\bibinfo  {journal}
  {arXiv:1509.08682}\ } (\bibinfo {year} {2015})}\BibitemShut {NoStop}%
\bibitem [{\citenamefont {Bandyopadhyay}\ and\ \citenamefont
  {Nathanson}(2013)}]{SeparableMeasurementsNathanson}%
  \BibitemOpen
  \bibfield  {author} {\bibinfo {author} {\bibfnamefont {S.}~\bibnamefont
  {Bandyopadhyay}}\ and\ \bibinfo {author} {\bibfnamefont {M.}~\bibnamefont
  {Nathanson}},\ }\bibfield  {title} {\enquote {\bibinfo {title} {Tight bounds
  on the distinguishability of quantum states under separable measurements},}\
  }\href {\doibase 10.1103/PhysRevA.88.052313} {\bibfield  {journal} {\bibinfo
  {journal} {Phys. Rev. A}\ }\textbf {\bibinfo {volume} {88}},\ \bibinfo
  {pages} {052313} (\bibinfo {year} {2013})}\BibitemShut {NoStop}%
\bibitem [{\citenamefont {{Bandyopadhyay}}\ \emph {et~al.}(2014)\citenamefont
  {{Bandyopadhyay}}, \citenamefont {{Cosentino}}, \citenamefont {{Johnston}},
  \citenamefont {{Russo}}, \citenamefont {{Watrous}},\ and\ \citenamefont
  {{Yu}}}]{SeparableMeasurementsRusso}%
  \BibitemOpen
  \bibfield  {author} {\bibinfo {author} {\bibfnamefont {S.}~\bibnamefont
  {{Bandyopadhyay}}}, \bibinfo {author} {\bibfnamefont {A.}~\bibnamefont
  {{Cosentino}}}, \bibinfo {author} {\bibfnamefont {N.}~\bibnamefont
  {{Johnston}}}, \bibinfo {author} {\bibfnamefont {V.}~\bibnamefont {{Russo}}},
  \bibinfo {author} {\bibfnamefont {J.}~\bibnamefont {{Watrous}}}, \ and\
  \bibinfo {author} {\bibfnamefont {N.}~\bibnamefont {{Yu}}},\ }\bibfield
  {title} {\enquote {\bibinfo {title} {{Limitations on separable measurements
  by convex optimization}},}\ }\href {http://arxiv.org/abs/1408.6981}
  {\bibfield  {journal} {\bibinfo  {journal} {arXiv:1408.6981}\ } (\bibinfo
  {year} {2014})}\BibitemShut {NoStop}%
\bibitem [{\citenamefont {Baez}()}]{BaezDirectSum}%
  \BibitemOpen
  \bibfield  {author} {\bibinfo {author} {\bibfnamefont {J.}~\bibnamefont
  {Baez}},\ }\href {http://math.ucr.edu/home/baez/photon/tensor.htm} {\enquote
  {\bibinfo {title} {Tensor product and direct sum},}\ }\bibinfo {note} {See
  also
  \href[pdfnewwindow]{https://proofwiki.org/wiki/Definition:Hilbert_Space_Direct_Sum}{ProofWiki:Hilbert\_Space\_Direct\_Sum}}\BibitemShut
  {NoStop}%
\bibitem [{\citenamefont {Hunter}\ and\ \citenamefont
  {Nachtergaele}(2001)}]{BookDirectSum}%
  \BibitemOpen
  \bibfield  {author} {\bibinfo {author} {\bibfnamefont {J.~K.}\ \bibnamefont
  {Hunter}}\ and\ \bibinfo {author} {\bibfnamefont {B.}~\bibnamefont
  {Nachtergaele}},\ }\href {\doibase 10.1142/4319} {\emph {\bibinfo {title}
  {Applied Analysis}}}\ (\bibinfo  {publisher} {World Scientific},\ \bibinfo
  {year} {2001})\BibitemShut {NoStop}%
\bibitem [{\citenamefont {Hanner}\ and\ \citenamefont
  {Rådström}(1951)}]{Hanner1951}%
  \BibitemOpen
  \bibfield  {author} {\bibinfo {author} {\bibfnamefont {O.}~\bibnamefont
  {Hanner}}\ and\ \bibinfo {author} {\bibfnamefont {H.}~\bibnamefont
  {Rådström}},\ }\bibfield  {title} {\enquote {\bibinfo {title} {A
  generalization of a theorem of {F}enchel},}\ }\href {\doibase
  10.1090/s0002-9939-1951-0044142-0} {\bibfield  {journal} {\bibinfo  {journal}
  {Proc. Amer. Math. Soc.}\ }\textbf {\bibinfo {volume} {2}},\ \bibinfo {pages}
  {589} (\bibinfo {year} {1951})}\BibitemShut {NoStop}%
\bibitem [{\citenamefont {Bárány}\ and\ \citenamefont
  {Karasev}(2012)}]{Brny2012}%
  \BibitemOpen
  \bibfield  {author} {\bibinfo {author} {\bibfnamefont {I.}~\bibnamefont
  {Bárány}}\ and\ \bibinfo {author} {\bibfnamefont {R.}~\bibnamefont
  {Karasev}},\ }\bibfield  {title} {\enquote {\bibinfo {title} {Notes about the
  {C}arathéodory number},}\ }\href {\doibase 10.1007/s00454-012-9439-z}
  {\bibfield  {journal} {\bibinfo  {journal} {Disc. Comp. Geom.}\ }\textbf
  {\bibinfo {volume} {48}},\ \bibinfo {pages} {783} (\bibinfo {year}
  {2012})}\BibitemShut {NoStop}%
\bibitem [{\citenamefont {Masanes}(2005)}]{MasanesQubitsEarly}%
  \BibitemOpen
  \bibfield  {author} {\bibinfo {author} {\bibfnamefont {L.}~\bibnamefont
  {Masanes}},\ }\bibfield  {title} {\enquote {\bibinfo {title} {{Extremal
  quantum correlations for N parties with two dichotomic observables per
  site}},}\ }\href {http://arxiv.org/abs/quant-ph/0512100} {\bibfield
  {journal} {\bibinfo  {journal} {quant-ph/0512100}\ } (\bibinfo {year}
  {2005})}\BibitemShut {NoStop}%
\bibitem [{\citenamefont {Masanes}(2006)}]{MasanesQubits}%
  \BibitemOpen
  \bibfield  {author} {\bibinfo {author} {\bibfnamefont {L.}~\bibnamefont
  {Masanes}},\ }\bibfield  {title} {\enquote {\bibinfo {title} {{Asymptotic
  violation of {B}ell inequalities and distillability}},}\ }\href {\doibase
  10.1103/PhysRevLett.97.050503} {\bibfield  {journal} {\bibinfo  {journal}
  {Phys. Rev. Lett.}\ }\textbf {\bibinfo {volume} {97}},\ \bibinfo {pages}
  {050503} (\bibinfo {year} {2006})}\BibitemShut {NoStop}%
\bibitem [{\citenamefont {Fritz}(2012{\natexlab{b}})}]{FritzDuality}%
  \BibitemOpen
  \bibfield  {author} {\bibinfo {author} {\bibfnamefont {T.}~\bibnamefont
  {Fritz}},\ }\bibfield  {title} {\enquote {\bibinfo {title} {Polyhedral
  duality in {B}ell scenarios with two binary observables},}\ }\href {\doibase
  10.1063/1.4734586} {\bibfield  {journal} {\bibinfo  {journal} {J. Math.
  Phys.}\ }\textbf {\bibinfo {volume} {53}},\ \bibinfo {eid} {072202} (\bibinfo
  {year} {2012}{\natexlab{b}}),\ 10.1063/1.4734586}\BibitemShut {NoStop}%
\bibitem [{\citenamefont {Collins}\ and\ \citenamefont
  {Gisin}(2004)}]{I3322Original}%
  \BibitemOpen
  \bibfield  {author} {\bibinfo {author} {\bibfnamefont {D.}~\bibnamefont
  {Collins}}\ and\ \bibinfo {author} {\bibfnamefont {N.}~\bibnamefont
  {Gisin}},\ }\bibfield  {title} {\enquote {\bibinfo {title} {{A Relevant Two
  Qubit Bell Inequality Inequivalent to the CHSH Inequality}},}\ }\href
  {http://stacks.iop.org/0305-4470/37/i=5/a=021} {\bibfield  {journal}
  {\bibinfo  {journal} {J. Phys. A}\ }\textbf {\bibinfo {volume} {37}},\
  \bibinfo {pages} {1775} (\bibinfo {year} {2004})}\BibitemShut {NoStop}%
\bibitem [{\citenamefont {Pál}\ and\ \citenamefont
  {Vértesi}(2010)}]{I3322NPA2}%
  \BibitemOpen
  \bibfield  {author} {\bibinfo {author} {\bibfnamefont {K.~F.}\ \bibnamefont
  {Pál}}\ and\ \bibinfo {author} {\bibfnamefont {T.}~\bibnamefont
  {Vértesi}},\ }\bibfield  {title} {\enquote {\bibinfo {title} {{Maximal
  violation of a bipartite three-setting, two-outcome {B}ell inequality using
  infinite-dimensional quantum systems}},}\ }\href {\doibase
  10.1103/PhysRevA.82.022116} {\bibfield  {journal} {\bibinfo  {journal} {Phys.
  Rev. A}\ }\textbf {\bibinfo {volume} {82}},\ \bibinfo {pages} {022116}
  (\bibinfo {year} {2010})}\BibitemShut {NoStop}%
\bibitem [{\citenamefont {Liang}\ \emph {et~al.}(2011)\citenamefont {Liang},
  \citenamefont {Vértesi},\ and\ \citenamefont
  {Brunner}}]{InfiniteDimensionViolation}%
  \BibitemOpen
  \bibfield  {author} {\bibinfo {author} {\bibfnamefont {Y.-C.}\ \bibnamefont
  {Liang}}, \bibinfo {author} {\bibfnamefont {T.}~\bibnamefont {Vértesi}}, \
  and\ \bibinfo {author} {\bibfnamefont {N.}~\bibnamefont {Brunner}},\
  }\bibfield  {title} {\enquote {\bibinfo {title} {Semi-device-independent
  bounds on entanglement},}\ }\href {\doibase 10.1103/PhysRevA.83.022108}
  {\bibfield  {journal} {\bibinfo  {journal} {Phys. Rev. A}\ }\textbf {\bibinfo
  {volume} {83}},\ \bibinfo {pages} {022108} (\bibinfo {year}
  {2011})}\BibitemShut {NoStop}%
\bibitem [{\citenamefont {{Palazuelos}}\ and\ \citenamefont
  {{Yin}}(2015)}]{DimensionDependantBellViolation3}%
  \BibitemOpen
  \bibfield  {author} {\bibinfo {author} {\bibfnamefont {C.}~\bibnamefont
  {{Palazuelos}}}\ and\ \bibinfo {author} {\bibfnamefont {Z.}~\bibnamefont
  {{Yin}}},\ }\bibfield  {title} {\enquote {\bibinfo {title} {{Large bipartite
  Bell violations with dichotomic measurements}},}\ }\href
  {http://arxiv.org/abs/1504.05769} {\bibfield  {journal} {\bibinfo  {journal}
  {arXiv:1504.05769}\ } (\bibinfo {year} {2015})}\BibitemShut {NoStop}%
\bibitem [{\citenamefont {Collins}\ \emph {et~al.}(2002)\citenamefont
  {Collins}, \citenamefont {Gisin}, \citenamefont {Linden}, \citenamefont
  {Massar},\ and\ \citenamefont {Popescu}}]{CGLMP02}%
  \BibitemOpen
  \bibfield  {author} {\bibinfo {author} {\bibfnamefont {D.}~\bibnamefont
  {Collins}}, \bibinfo {author} {\bibfnamefont {N.}~\bibnamefont {Gisin}},
  \bibinfo {author} {\bibfnamefont {N.}~\bibnamefont {Linden}}, \bibinfo
  {author} {\bibfnamefont {S.}~\bibnamefont {Massar}}, \ and\ \bibinfo {author}
  {\bibfnamefont {S.}~\bibnamefont {Popescu}},\ }\bibfield  {title} {\enquote
  {\bibinfo {title} {{Bell inequalities for arbitrarily high-dimensional
  systems}},}\ }\href {\doibase 10.1103/PhysRevLett.88.040404} {\bibfield
  {journal} {\bibinfo  {journal} {Phys. Rev. Lett.}\ }\textbf {\bibinfo
  {volume} {88}},\ \bibinfo {pages} {040404} (\bibinfo {year}
  {2002})}\BibitemShut {NoStop}%
\bibitem [{\citenamefont {Zohren}\ and\ \citenamefont {Gill}(2008)}]{CGLMP08}%
  \BibitemOpen
  \bibfield  {author} {\bibinfo {author} {\bibfnamefont {S.}~\bibnamefont
  {Zohren}}\ and\ \bibinfo {author} {\bibfnamefont {R.~D.}\ \bibnamefont
  {Gill}},\ }\bibfield  {title} {\enquote {\bibinfo {title} {{Maximal violation
  of the Collins-Gisin-Linden-Massar-Popescu inequality for infinite
  dimensional states}},}\ }\href {\doibase 10.1103/PhysRevLett.100.120406}
  {\bibfield  {journal} {\bibinfo  {journal} {Phys. Rev. Lett.}\ }\textbf
  {\bibinfo {volume} {100}},\ \bibinfo {pages} {120406} (\bibinfo {year}
  {2008})}\BibitemShut {NoStop}%
\bibitem [{\citenamefont {Acín}\ \emph {et~al.}(2012)\citenamefont {Acín},
  \citenamefont {Massar},\ and\ \citenamefont
  {Pironio}}]{RandomnessVNonlocality}%
  \BibitemOpen
  \bibfield  {author} {\bibinfo {author} {\bibfnamefont {A.}~\bibnamefont
  {Acín}}, \bibinfo {author} {\bibfnamefont {S.}~\bibnamefont {Massar}}, \
  and\ \bibinfo {author} {\bibfnamefont {S.}~\bibnamefont {Pironio}},\
  }\bibfield  {title} {\enquote {\bibinfo {title} {Randomness versus
  nonlocality and entanglement},}\ }\href {\doibase
  10.1103/PhysRevLett.108.100402} {\bibfield  {journal} {\bibinfo  {journal}
  {Phys. Rev. Lett.}\ }\textbf {\bibinfo {volume} {108}},\ \bibinfo {pages}
  {100402} (\bibinfo {year} {2012})}\BibitemShut {NoStop}%
\bibitem [{\citenamefont {{Methot}}\ and\ \citenamefont
  {{Scarani}}(2007)}]{NonlocalityAnomaly}%
  \BibitemOpen
  \bibfield  {author} {\bibinfo {author} {\bibfnamefont {A.~A.}\ \bibnamefont
  {{Methot}}}\ and\ \bibinfo {author} {\bibfnamefont {V.}~\bibnamefont
  {{Scarani}}},\ }\bibfield  {title} {\enquote {\bibinfo {title} {{An anomaly
  of non-locality}},}\ }\href {http://arxiv.org/abs/quant-ph/0601210}
  {\bibfield  {journal} {\bibinfo  {journal} {Quant. Info. Comp.}\ }\textbf
  {\bibinfo {volume} {7}},\ \bibinfo {pages} {157} (\bibinfo {year}
  {2007})}\BibitemShut {NoStop}%
\bibitem [{\citenamefont {{Plastino}}\ \emph {et~al.}(2015)\citenamefont
  {{Plastino}}, \citenamefont {{Bellomo}},\ and\ \citenamefont
  {{Plastino}}}]{DimensionAsResource}%
  \BibitemOpen
  \bibfield  {author} {\bibinfo {author} {\bibfnamefont {A.}~\bibnamefont
  {{Plastino}}}, \bibinfo {author} {\bibfnamefont {G.}~\bibnamefont
  {{Bellomo}}}, \ and\ \bibinfo {author} {\bibfnamefont {A.~R.}\ \bibnamefont
  {{Plastino}}},\ }\bibfield  {title} {\enquote {\bibinfo {title} {{Quantum
  state space-dimension as a quantum resource}},}\ }\href
  {http://arxiv.org/abs/1505.05455} {\bibfield  {journal} {\bibinfo  {journal}
  {arXiv:1505.05455}\ } (\bibinfo {year} {2015})}\BibitemShut {NoStop}%
\bibitem [{\citenamefont {Chaves}\ \emph {et~al.}(2014)\citenamefont {Chaves},
  \citenamefont {Luft},\ and\ \citenamefont {Gross}}]{GeometryAndNovel}%
  \BibitemOpen
  \bibfield  {author} {\bibinfo {author} {\bibfnamefont {R.}~\bibnamefont
  {Chaves}}, \bibinfo {author} {\bibfnamefont {L.}~\bibnamefont {Luft}}, \ and\
  \bibinfo {author} {\bibfnamefont {D.}~\bibnamefont {Gross}},\ }\bibfield
  {title} {\enquote {\bibinfo {title} {Causal structures from entropic
  information: geometry and novel scenarios},}\ }\href {\doibase
  10.1088/1367-2630/16/4/043001} {\bibfield  {journal} {\bibinfo  {journal}
  {New J. Phys.}\ }\textbf {\bibinfo {volume} {16}},\ \bibinfo {pages} {043001}
  (\bibinfo {year} {2014})}\BibitemShut {NoStop}%
\end{thebibliography}%

\end{document}